\documentclass[journal,12pt,onecolumn]{IEEEtran}

\usepackage[pdftex]{graphicx}

\usepackage{fullpage}
\usepackage{cite}
\usepackage[cmex10]{amsmath}
\usepackage{algorithmic}
\usepackage{mdwmath}
\usepackage{mdwtab}
\usepackage[font=footnotesize]{subfig}

\newtheorem{theorem}{Theorem}[section]

\newenvironment{proof}[1][Proof]{\begin{trivlist}
\item[\hskip \labelsep {\bfseries #1}]}{\end{trivlist}}

\newcommand{\qed}{\nobreak \ifvmode \relax \else
      \ifdim\lastskip<1.5em \hskip-\lastskip
      \hskip1.5em plus0em minus0.5em \fi \nobreak
      \vrule height0.75em width0.5em depth0.25em\fi}

\begin{document}
%
\title{On Improving the Representation of a Region\\ Achieved by a Sensor Network}
\author{\IEEEauthorblockN{Xiaoyu Chu and Harish Sethu}\\
\IEEEauthorblockA{Department of Electrical and Computer Engineering\\
Drexel University\\
Philadelphia, PA 19104-2875\\
Email: \{xiaoyu.chu, sethu\}@drexel.edu}
}

\maketitle
\thispagestyle{empty}
~\vskip 0.5in
\begin{abstract}
This report considers the class of applications of sensor networks in
which each sensor node makes measurements, such as temperature or
humidity, at the precise location of the node. Such {\em spot-sensing}
applications approximate the physical condition of the entire region of
interest by the measurements made at only the points where the
sensor nodes are located. Given a certain density of nodes in a
region, a more spatially uniform distribution of the nodes leads to a
better approximation of the physical condition of the region. This
report considers the error in this approximation and seeks to improve
the quality of representation of the physical condition of the points
in the region in the data collected by the sensor network. We develop
two essential metrics which together allow a rigorous quantitative assessment of
the quality of representation achieved: the average representation
error and the unevenness of representation error, the latter based on
a well-accepted measure of inequality used in economics. We present the
rationale behind the use of these metrics and derive relevant
theoretical bounds on them in the common scenario of a planar region
of arbitrary shape covered by a sensor network deployment. A simple
new heuristic algorithm is presented for each node to determine if
and when it should sense or sleep to conserve energy while also
preserving the quality of representation. Simulation results show that
it achieves a significant improvement in the quality of representation
compared to other related distributed algorithms. Interestingly,
our results also show that improved spatial uniformity has the welcome
side-effect of a significant increase in the network lifetime.\footnote{A
  preliminary version of this manuscript appeared in {\em Proceedings
    of IEEE INFOCOM 2009}. This research was partially funded by NSF
  Award CNS-0626548.}
\end{abstract}

\newpage
\section{Introduction}
\label{sec:intro}

Networks of inexpensive low-power sensor nodes may be deployed to
sense, gather and process information in a region of interest for a
variety of purposes including surveillance, target tracking,
wildlife monitoring and pollution studies \cite{AkySu2002}. Based on
the expected behavior of individual nodes, these applications of
sensor networks may be broadly categorized into two types: {\em
  area-sensing} applications and {\em spot-sensing}
applications. Examples of area-sensing applications include enemy
surveillance, target tracking, intrusion detection and wildlife
monitoring through audio/image/video recording; in these applications,
sensor nodes make relevant observations within a local {\em sensing
  area} using vision, sound, seismic-acoustic energy, infrared energy, or magnetic
field changes. On the other hand, in spot-sensing applications, each
sensor node makes measurements of physical phenomena such as
temperature, humidity and environmental pollution at precisely the
spot where it is located, and there is no concept of a sensing area.
The physical condition of each point in the region of interest is
represented in the data collected from nearby active sensor nodes.
The farther the nearest active nodes are from a point, the poorer is
the representation of the physical condition at the point in the
data collected by the sensor network. For example, if most of the
active sensor nodes are clustered together in one corner of a
region, the quality of representation of the region is likely to be
poor. A more spatially uniform distribution, however, will lead to
an improved quality of representation. In this report, we consider
spot-sensing applications and introduce the problem of improving
this quality of representation in the data collected by the sensor
network.

The problem of improving the quality of representation is related but
different from the coverage problems typically considered for area-sensing
applications
\cite{MegFar2001,HuaTse2005,CarGra2006,LiWan2003,GalCar2008,HuaLo2006,ShaSri2003}. In
most of these coverage problems, a region is considered $k$-covered if
all points in it are within the sensing area of at least $k$ active
nodes. Such a notion of coverage, while appropriate for area-sensing
applications, is not relevant for spot-sensing applications where
there is no concept of a sensing area. In spot-sensing
applications, the quality of representation enjoyed by a point in the
region depends on the desired spatial granularity with which the physical
condition needs to be sampled and on some function of the distances to
the nearest set of active sensor nodes around the point. This report
introduces new metrics that help evaluate the quality of
representation achieved by a sensor network deployed in a
region.

Energy being a key constraint in most sensor networks, this work
assumes that a sensor node can be programmed to make a choice at
specific intervals of time on whether it should be in the {\em
sense} mode (also referred in this report interchangeably as the {\em
active} mode) or the sleep mode (in which its sensing module is
turned off). An additional goal in spot-sensing applications becomes
one of developing a distributed algorithm to determine sleep/sense
times with the specific goals of (i) conserving energy, (ii)
achieving the desired spatial granularity with which the physical
condition in the region is sampled by achieving the appropriate
spatial density of active nodes, and (iii) finally, achieving a high
quality of representation of the region at all times by the network
of active nodes.

The metrics for the quality of representation and the above goals of a
distributed algorithm are also relevant in the context of sensor
networks with transducer heterogeneity. It is becoming
increasingly common in real world sensor network applications to
integrate data from several different types of transducers
\cite{Cro,IEEE1451}. Microsensors, especially those using
microelectromechanical systems (MEMS), permit the sensing of a
variety of physical phenomena on a single sensor node
\cite{CulEst2004}. Sensor nodes such as the Berkeley MICA Mote
typically integrate several transducer types, such as for
acceleration, temperature, light and sound, on a single board
\cite{HilSze2000}. Each sensor node typically has dynamic control
over which transducers are active. Since different physical
phenomena generally require sensing at different spatial
granularities, one can avoid unnecessary energy consumption by
activating only a subset of transducers at each of the sensor nodes.
This calls for distributed algorithms executed by all the sensor
nodes to automate the process of determining which transducers
should be activated on which nodes based on the desired density of
each transducer type while also ensuring that all points in the
region are well-represented by measurements made by each transducer
type.

\subsection{Problem Statement}\label{sec:ProblemStatement}

Consider $N$ sensor nodes distributed within a certain region of
interest, denoted by $R$. Let $G^\prime=(V^\prime,E^\prime)$ denote
the graph of these sensor nodes where each node $u \in V^\prime$
represents a sensor node and each edge $(u,v) \in E^\prime$
represents the fact that nodes $u$ and $v$ are neighbors and can
communicate directly with each other. Let $d(u,v)$ denote the
Euclidean distance between sensor nodes represented by vertices $u$
and $v$. The Euclidean distances between nodes are computable if the
nodes are all fitted with low-power GPS receivers, or through
location estimation techniques if only a subset of nodes are
equipped with GPS receivers \cite{NicNat2003}, or by estimating
distances based on exchanging transmission and reception powers
\cite{KriBie1997}.

Let $z$ denote the {\em desired spatial density}, in number of
active nodes per unit area, determined based on the spatial
granularity with which the physical phenomenon of interest should be
sensed. Let $G = (V, E)$ denote the subgraph of $G^\prime$ such that
$v \in V$ iff vertex $v$ represents an active node (as opposed to
one in sleep mode). In this report, we do not require that $G$ be a
connected graph (because, in many applications, retrieval of data from
sensor nodes may be accomplished through mobile gateways
\cite{RamChe2008}). The problem now is one
of determining $G$ in a distributed manner so that every point in
the region is well-represented by the active sensor nodes.

There are two key aspects to this problem:
\begin{enumerate}
\item What are the metrics that one should use to measure
  the quality of representation achieved by $G$?
\item Given the metrics for quality of representation, what is a distributed
  algorithm that one should employ to determine $G$ (i.e., to
  determine which nodes should sense and which nodes should sleep)
  while achieving a high quality of representation?
\end{enumerate}

\subsection{Contributions and Organization}
We propose a new problem described in the previous subsection
specifically for spot-sensing applications in sensor networks. We
develop a pair of metrics that together allow a quantitative
assessment of the quality of representation:
the {\em average representation error} of the points in the
region and the {\em unevenness of representation error}
across the points in the region. Section~\ref{sec:metrics} presents
these metrics along with the rationale behind them. Based on the
average representation error of the points in the region,
Section~\ref{sec:metricD} develops a metric normalized by the desired
spatial density to allow for comparative evaluations of the quality of
representation achieved by a network across different desired spatial
densities. Section~\ref{sec:metricU} borrows from the field of
economics and uses the Gini index, a well-accepted measure of
inequality, to develop a new metric for the
unevenness of representation error among the points in the
region. Section~\ref{sec:WhyTwoMetrics} discusses the need for both of
these complementary metrics. Lower bounds on both metrics are derived in
Section~\ref{sec:lowerBounds}. Upper bounds on these lower bounds for
the common scenario of a continuous two-dimensional region covered by
a sensor network are derived in Appendix~\ref{sec:uniform_bounds}.
Section~\ref{relatedWork} discusses work in sensor networks as
well as in other fields which seek to solve similar or related underlying
mathematical problems.

Section~\ref{sec:algorithm} develops a generalized, distributed
algorithm, called EvenRep($\cal{F}$,\,$L$), to achieve a better
quality of representation for points in the region of interest. The
algorithm, a simple heuristic, is parametrized by two quantities:
$\cal{F}$, a {\em target distance function} which specifies the
desired distance between an active node and its $k$-th nearest active
neighbor and $L$, the maximum number of active neighbors that a node
should consider in making its decision to sense or sleep. The
algorithm seeks to achieve the target distance function for all active
nodes. The target distance function, $\cal{F}$, can depend on whether
the region of interest is
2-dimensional or 3-dimensional, the type of application, or any
spatial constraints specific to the region. Section~\ref{evenCover}
describes the pseudo-code of the algorithm and the rationale behind
it. We find that the ideal target distance function is one based on
the region of interest being tessellated by congruent hexagonal cells
with an active sensor node at the center of each cell. We denote this
target distance function by $H$ and Section~\ref{EvenCoverH} describes
EvenRep($H, L$), used in our simulation results.

Section~\ref{results} presents several simulation results on the
performance of EvenRep($H,L$) and some other representative
algorithms. The results show that the EvenRep($H,L$) algorithm
achieves a significant improvement in the quality of representation in
comparison to other algorithms. We show that achieving an improved
quality of representation has a welcome side-effect of significantly
improving the network lifetime. In fact, we show that
EvenRep($H, L$) achieves almost a $50\%$ increase in the network
lifetime in comparison to other related distributed
algorithms. Section~\ref{conclusion} concludes the report.

\section{The Metrics}
\label{sec:metrics}

In this section, we develop metrics to quantify the quality of
representation achieved in a sensor network deployment for
spot-sensing applications. Past research that discusses related
metrics has largely assumed a system model that is more appropriate
for coverage problems in area-sensing applications
\cite{MegFar2001,HuaTse2005,CarGra2006,LiWan2003,GalCar2008,HuaLo2006,ShaSri2003}. In
these problems, each sensor node has a pre-defined sensing range
and the goal is to ensure that each point in the region of interest is
$k$-covered, i.e., lies within the sensing area of at least
$k$ active sensor nodes. As opposed to coverage at a point in an area-sensing
application, the quality of representation of a point in a spot-sensing
application is not easily captured in an either-or binary manner, an
implication of the fact that there is no concept of a sensing area in
spot-sensing applications. Even modified coverage problems for
area-sensing applications, such as when a point is considered either
covered or uncovered with a probability that is a function of the
distance to the nearest sensor node \cite{WanTse2008}, do end up
imposing a binary either-or assessment that is not useful to assessing
the quality of representation of the point. Also, metrics based on the
distances between active nodes (e.g., \cite{ClaEva1954}) used in
solving different problems do not capture the quality of
representation for spot-sensing applications, a quantity that is more
about the points in the region of interest than the distances between
neighboring active nodes.

The quality of representation of a point depends on the error in the
representation of the point in the data collected by the sensor
nodes. As mentioned in Section~\ref{sec:intro}, this error depends on
some function of the distances between the point and the nearest
active nodes. This function may be different for different physical
conditions and is sometimes known (as discussed in \cite{KraGue2006})
but, most often, is unknown before network deployment. For
clarity of presentation, we describe our work using the
case in which the error in the representation of a point may be
assumed to be directly proportional to the distance between the point
and its nearest active sensor node (the error is zero if there is an active
sensor node exactly at that point). However, the metrics of quality of
representation that we develop can be readily adapted to other cases
with different relationships between representation error at a point
and the distances to the nearby active nodes. Further, the heuristic
algorithm we present later in this report is also independent of this
assumption. Also for clarity, we present this work assuming
that the region covered by the sensor network is a $2$-dimensional
plane. The metrics presented here and the algorithm can be readily
adapted to the $3$-dimensional case.

Thus, the average of representation errors at all points in the
region, normalized by the desired spatial granularity of active nodes
for the physical condition being sampled, is one aspect of the quality
of representation of the region. However, as we will show later in
this section, a low average representation error alone does not tell
the whole story and that an even spread of these values is also an
essential aspect of the quality of representation. In the following,
we formalize and develop a rigorous definition of two metrics: the
{\em average representation error} based on the normalized average of
the distances of the points to their respective nearest active nodes,
and the {\em unevenness of representation error} based on the
distribution of these distances.

\subsection{Average Representation Error}
\label{sec:metricD}

Let $d_p(v)$ denote the distance of node $v$
from point $p$. Let $n_p(G)$ denote the nearest node in $G$ (the set
of active nodes) from point $p$. Let $\overline{d_R(G)}$ denote the average
value of $d_p(n_p(G))$ over all points $p$ in region
$R$. $\overline{d_R(G)}$ may also be thought of as the expected value
of $d_p(n_p(G))$ for a random point $p$ in the region. Intuitively,
given the same area of the region of interest and the same number of
active sensor nodes, the smaller the value of $\overline{d_R(G)}$
the lower the representation error. However, as mentioned in
Section~\ref{sec:intro}, the representation error should also depend on the desired
spatial density of active nodes required for sampling of
the physical phenomenon at the point (for example, particulate
pollution may have to be sampled at a higher spatial granularity than
temperature and so, the same average distances may not imply the same
representation error).  Different applications can tolerate different
average distances between points and the nearest active nodes for the
same quality of representation; therefore, without knowledge of the
desired spatial density, $z$, the value of $\overline{d_R(G)}$ reveals
little about the quality of representation achieved for an
application. Therefore, an appropriate metric is one that uses the
average distance, $\overline{d_R(G)}$, normalized by the average
distance in the best-case scenario at the desired spatial density.

\begin{figure}[!t]
\begin{center}
    \subfloat[{One active node at the center of a unit square area.}]{
      \hskip0.3in
      \includegraphics[height=0.7in]{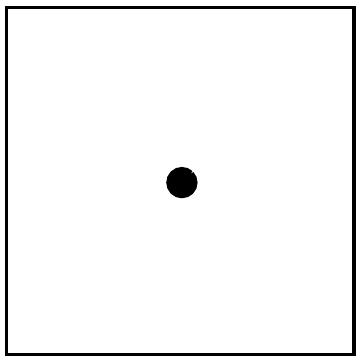}
      \hskip0.3in
        \label{fig:orignal}
        }\hskip0.3in
    \subfloat[{An active node placed at the center of each of the four
      quarters of a unit square area.}]{
      \hskip0.3in
        \includegraphics[height=0.7in]{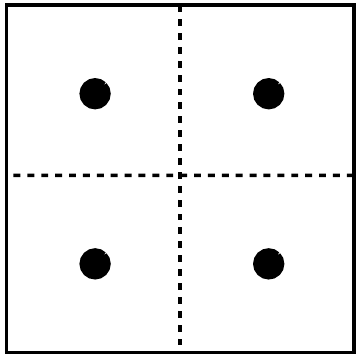}
      \hskip0.3in
        \label{fig:four_times}
        }
    \caption{An example to illustrate the average representation error
      as a metric; $D(G,R)$ is the same in the two cases when the
      desired spatial density in (a) is 1 but in (b) is 4.}
    \label{fig:avg}
\end{center}
\end{figure}

The best-case scenario occurs when the region of interest can be
covered in a space-filling fashion by non-overlapping circular areas
with an active sensor node at the center of each circular area. Note
that such a scenario is not realistic and is used here only as a
means to derive a normalization factor in the metric. Given a
desired spatial density of $z$, the radius of these circular areas
in the best-case scenario is given by $r= 1/\sqrt{z\pi}$ (recall
that $z$ denotes the desired number of active nodes per unit area,
the size of each circular area is $\pi r^2$, and therefore, $z =
1/\pi r^2$). The expected distance from points in the region to the
nearest sensor node in this case is given by:
\[
\int_0^r \frac{2\pi x}{\pi r^2} x dx =
\frac{2r}{3}=\frac{2}{3\sqrt{z\pi}}
\]
Thus, the normalized expected distance of points to their respective
nearest nodes is given by:
\[
\overline{d_R(G)} \left( {\displaystyle \frac{3\sqrt{z\pi}}{2}}
\right)
\]
Dispensing with the constant, $3\sqrt{\pi}/2$, we
define the {\em average representation error}, denoted by $D(G,R)$, as:
\begin{equation}
D(G,R)=\overline{d_R(G)} \, \sqrt{z}
\label{equ:strength of coverage}
\end{equation}

Fig.\,\,\ref{fig:avg} is illustrative of the average representation
error as a metric. Consider a square region of
interest of unit area. The average representation error in
Fig.\,\,\ref{fig:orignal} when the desired spatial density is 1 is
the same as the average representation error when the desired
spatial density is 4.

\subsection{Unevenness of Representation Error}
\label{sec:metricU}

The field of economics has a long history of
measuring inequality and a vast body of literature on the topic
\cite{Cow1977,MarOlk1979}. For measuring the unevenness of
representation error, we use a popular and well-accepted metric in
economics, the Gini index, based on the relative mean difference
between the quantities being compared (in our case, the quantities
are distances of points to their respective nearest active sensor
nodes). Consider $m$ quantities, $g_1 \leq g_2 \leq \dots \leq g_m$.
The mean difference between these quantities is:
\[
\Delta = \frac{1}{m^2} \sum_{i=1}^m \sum_{j=1}^m | g_i - g_j |
\]
The relative mean difference is the mean difference divided by the
mean, $\overline{g}$. The Gini index is defined as one-half of the
relative mean difference, i.e.,
\[
{\rm Gini~index} = \frac{\Delta}{2\overline{g}} =
\frac{1}{2\overline{g}m^2} \sum_{i=1}^m \sum_{j=1}^m | g_i - g_j |
\]
Adapting the Gini index to the context of our problem poses
one issue: the number of quantities we have is infinite because of
the infinite number of points in any region of interest. Therefore,
instead of using summations, we consider expected values in defining
unevenness. Let $p$ and $q$ denote two arbitrary random points in the
region of interest $R$. We define the {\em unevenness of
  representation error}, $U(G, R)$, of graph
$G$ in the region $R$ as:
\begin{equation}
U(G, R) = \frac{E[ | d_p(n_p(G)) - d_q(n_q(G)) |
  ]}{2\overline{d_R(G)}}
\label{equ:giniIndex}
\end{equation}
where $p, q \in R$. The smaller the value of the above quantity, the
better the spatial uniformity.

\subsection{One Metric or Two?}
\label{sec:WhyTwoMetrics}

\begin{figure}[!t]
\begin{center}
\includegraphics[width=2in]{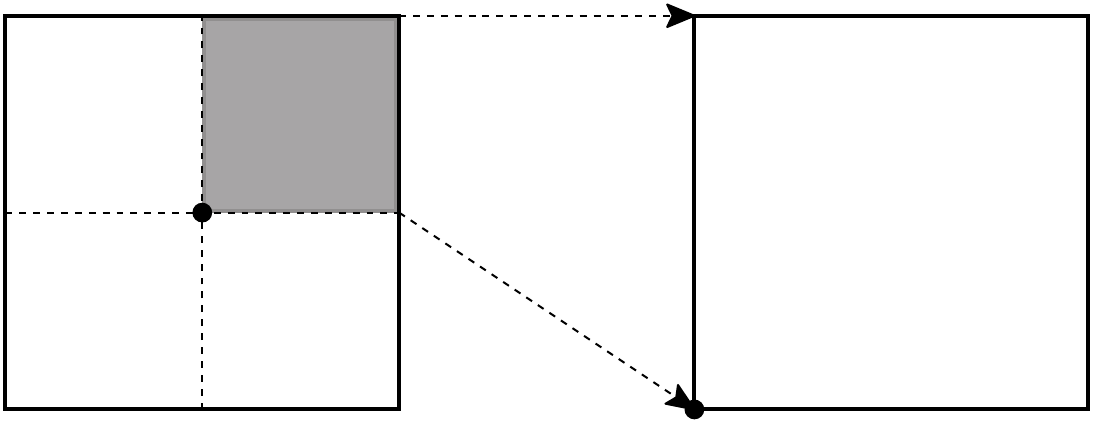}
\end{center}
\caption{An example illustrating that improving $U(G,R)$ does not necessarily improve $D(G,R)$.}
\label{fig:1}
\end{figure}

For both metrics, $D(G,R)$ and $U(G,R)$, a smaller value of the metric
implies better quality of representation. A legitimate question at
this point is whether minimizing one also minimizes the other, i.e.,
whether we need both of the above two metrics or if one of the
metrics above can serve as the sole metric for measuring the quality
of representation. We answer this using two simple examples: one in
which $U(G,R)$ is minimized but $D(G,R)$ is not; another in which
the $D(G,R)$ is minimized but $U(G,R)$ is not.

{\em Does improving the evenness also improve the average?}
Fig.\,\,\ref{fig:1} considers two regions, each of unit square area
but with different sensor node deployments for a desired spatial
density of 1 (i.e., we wish to place exactly one sensor node in the
region). In the first of the two regions, a sensor node is placed at
the centroid of the area while in the other region, the node is
placed at one of the corners of the square region. Note that both
$D(G,R)$ and $U(G,R)$ are minimized in the former case for the
desired spatial density of 1. Also note that $D(G,R)$ is maximized
in the latter case. We now claim that the unevenness of representation
is equal in the two cases and also the best achievable. The unevenness
of representation error depends on the normalized distribution of
the distances from the points in the region to the one active sensor
node. Dividing the square region of interest into four quarters as
shown in the region on the left-hand side in Fig.\,\,\ref{fig:1}, we
note that this distribution for points within each of the quarters
is identical to each other. Since the overall distribution of these
distances in the full square region of unit area is composed of the
identical distributions within each of the quarters, the unevenness of
representation in the region of unit square area in the left-hand
side region is the same as that within each of the quarters. Now,
the sensor node deployment shown in the region on the right-hand
side of Fig.\,\,\ref{fig:1} can be thought of as an enlargement of
one of the quarters in the region on the left-hand side, and
therefore, achieving the same degree of evenness. This example shows
that $U(G,R)$ can be the minimum possible when $D(G,R)$ is the
minimum or the maximum possible. This shows that achieving the
lowest possible unevenness of representation error does not
necessarily achieve the best average representation error; in fact,
sometimes it can even lead to the worst possible average
representation error.

\begin{figure}[!t]
\begin{center}
    \subfloat[{A placement in which $U(G,R)$ is minimized but $D(G,R)$
      is not.}]{
        \includegraphics[height=1.2in]{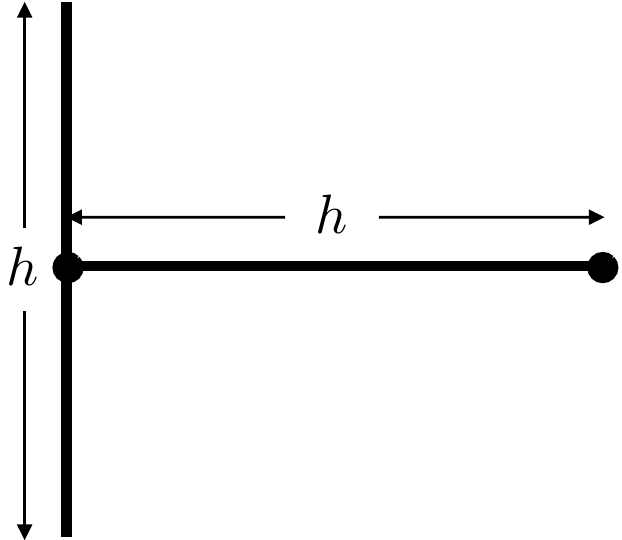}
        \label{fig:Ta}
        }\hskip0.25in
    \subfloat[{A placement in which $D(G,R)$ is minimized but $U(G,R)$ is not.}]{
        \includegraphics[height=1.2in]{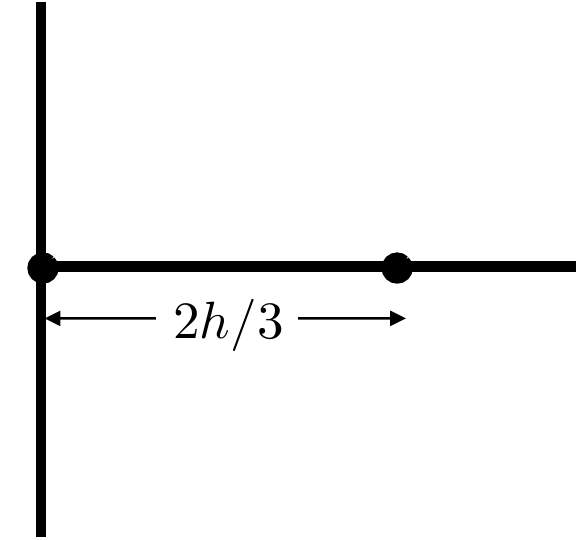}
        \label{fig:Tb}
        }
    \caption{An example illustrating that improving $D(G,R)$ does not necessarily improve $U(G,R)$.}
    \label{fig:T}
\end{center}
\end{figure}
{\em Does improving the average also improve the evenness?} Consider
a toy example of a region composed of two one-dimensional regions
arranged in the form of a sideways `T' as shown in
Fig.\,\,\ref{fig:T}. Assume that the desired spatial density
corresponds to placing two active sensor nodes in the region.
Fig.\,\,\ref{fig:Ta} illustrates a placement in which $U(G,R)$ is
minimized. Fig.\,\,\ref{fig:Tb} shows another placement of the two
sensor nodes in which $D(G,R)$ is minimized but $U(G,R)$ is not
minimized.

The two examples above show that improving the evenness does not
necessarily improve the average representation error and that
improving the average does not necessarily improve the evenness of
representation error. Therefore, both metrics are essential to
gaining insight into the quality of representation (though different
metrics may rank differently in their importance to different
applications).

\subsection{Lower Bounds on $D(G,R)$ and $U(G,R)$}
\label{sec:lowerBounds}

For all of our subsequent analysis, it is insightful to have values
of the lowest possible average representation error and the lowest
possible unevenness of representation error. The following theorem
proves these bounds.

\begin{theorem}
The lower bound on $D(G,R)$ is $2/(3\sqrt{\pi}) \approx 0.376$ and
the lower bound on $U(G,R)$ is 0.2.\label{theorem:lowerBound}
\end{theorem}
\begin{proof}
The lower bounds of both $D(G,R)$ and $U(G,R)$ are achieved when the
sensing nodes are perfectly evenly distributed such that the region
of interest is completely covered by non-overlapping circular areas
of radius $r$ with a sensing node at the center of each circle. Since each circle is identical to all
others as far as the distances of all points to their nearest active
nodes are concerned, $D(G,R)$ and $U(G,R)$ for each circular region
are the same as $D(G,R)$ and $U(G,R)$ for the entire region.

The lower bound on $D(G,R)$ is the average
distance between the center of a unit circle and the points within the
circle. This is an easily-derived geometric result
\cite{StoSto1994}.

Focusing now on $U(G,R)$, note that the average
distance between a point and the center of the circle is given by:
\begin{equation}
\overline{d} = \int_0^r \frac{2\pi x}{\pi r^2} x dx = \frac{2r}{3}
\label{equ:averaged}
\end{equation}
Consider two random
points within such a circle at distances $x$ and $y$ from the center
of the circle.
\begin{equation}
E[ |x-y|] = \int_0^r \frac{2\pi x}{\pi r^2} \int_0^r \frac{2\pi
y}{\pi
  r^2} |x-y| dy dx = \frac{4r}{15}
\label{equ:expectedD}
\end{equation}
Using Eqns.\,\,(\ref{equ:giniIndex}) and (\ref{equ:averaged}):
\[
U(G,R) = \left( \frac{1}{2} \right) \left( \frac{4r/15}{2r/3}
\right) = \frac{1}{5}
\]
\end{proof}

The bounds derived above are achieved in a scenario where the
region of interest can be perfectly covered by non-overlapping
circles. The simplest case would be a circular area with one sensor
node placed at the center. While possible, this is an unlikely
scenario for regions covered by a sensor network and therefore, in 
Appendix~\ref{sec:uniform_bounds}, we consider regions of interest
of arbitrary shape and prove upper bounds on the lower bounds of
$D(G,R)$ and $U(G,R)$. Theorem~\ref{theorm:achievablegini} in
Appendix~\ref{sec:uniform_bounds} shows that the lower bounds for
regions of arbitrary shape are only slightly higher than those for the
ideal case proved above.

\section{Related Work}\label{relatedWork}
The problem of designating the mode of a sensor
node as either active or sleeping is related (though not identical)
to the 2-color instance of some versions of the distributed graph
coloring problem \cite{Wes2001,OdoSet2004}, in which each node takes
on one of two colors with the goal to minimize the number of
neighbors of the same color as itself. While the algorithms in this
body of work will generally improve the spatial uniformity of active
nodes, they do not consider the distances between the nodes in their
computations and therefore, are limited in their application to the
problem under consideration. We show this later in
Section~\ref{results} by simulating the {\em Flip} algorithm, an
adaptation of the algorithm in \cite{Vaz2001}, in which each node
begins with randomly assigning itself one of two modes, and then, at
random intervals of time, switches to the mode that best
approximates the active node ratio in its neighborhood.

Spatially uniform distribution based on distances is more explicitly
considered in another body of work related to the problem of
facility location \cite{DreHam2004,OweDas1998}. The problem involves
determination of the locations of facilities in an environment (such
as emergency services in a city) given some constraints and an
objective function. In the field of networking, related problems
have been solved in the context of content distribution networks
where one has to replicate resources in multiple locations (servers
on a network) to boost performance by minimizing delay from users to
the nearest resource or by achieving load balancing on the network
\cite{QiuPad2001,LiGol1999,KimYoo2004,KanRob2002}. A variety of
techniques, including graph-theoretic approaches, heuristic
algorithms and dynamic programming, have been employed in these
works to arrive at a solution. Ko and Rubenstein developed the first
distributed algorithm for the placement of replicated resources,
best described as a solution to the distributed graph coloring
problem, by having each node continually change its color in a
greedy manner to maximize its own distance to a node of the same
color \cite{KoRub2005}. This work, which considers the distance
between two nodes as that along the communication path and not
as the geographical distance, cannot be directly applied to the
problem considered in this report. A further reason this body of work
does not directly apply here is that they only consider the
relationships between nodes and not between the nodes and the points
in the region of interest. Another set of works consider a set of
points in the region of interest as the targets, where the goal is
to cover and monitor each target point, as is discussed in
\cite{CarDu2005}. This set of work also does not serve the purpose
of achieving good quality of representation because, in our case,
all points in the region are equally significant targets.

Points in the region of interest are most explicitly considered in
the set of works that propose coverage algorithms for sensor
networks based on assuming a sensing area for each node in
area-sensing applications
\cite{MegFar2001,HuaTse2005,CarGra2006,LiWan2003,GalCar2008,HuaLo2006,ShaSri2003}.
The goal is usually to ensure that each point in the region of
interest is within the sensing area of at least $k$ active sensor
nodes. Distributed algorithms to achieve $k$-coverage do not
much improve the quality of representation although they do not
specifically attempt it either. In Section~\ref{results}, we
will compare the quality of sensor node representation and the
lifetime of a representative member of this class of algorithms with
the one proposed in this report.

\section{The EvenRep($\cal{F}$,\,$L$) Algorithm}\label{sec:algorithm}
The design of a distributed algorithm for the problem stated in
Section~\ref{sec:ProblemStatement} requires that a node make an
estimate of the quality of representation in its local area in
comparison to the desired spatial density to reliably determine if
it should sleep or go active. Note that even though the quality of
representation is about the distances of points in the region to the
nodes, only nodes and not the points can participate in this
distributed algorithm and so, we have to use heuristics based on
distances between nodes to achieve an improved quality of
representation for the points. A node, therefore, needs to know the
expected distances to the nearest active neighbors in a target
distance function and compare these against the actual distances.
Let $\cal{F}$ denote a {\em target distance function} which
specifies a mapping between $k \geq 1$ and the target distance
between an active node and its $k$-th nearest active neighbor. The
origin of a target distance function, say $\cal{F}$, may be the
expected distances between neighboring nodes in a given spatial
distribution $\cal{S}$, but targeting $\cal{F}$ in the algorithm is
not necessarily the same as targeting $\cal{S}$. For the same
reason, $\cal{F}$ is not necessarily a mapping between $k$ and the
expected distance between a node and its $k$-th nearest active
neighbor in the spatial distribution $\cal{S}$.
\begin{figure}[!t]
\begin{center}
\includegraphics[width=1.5in]{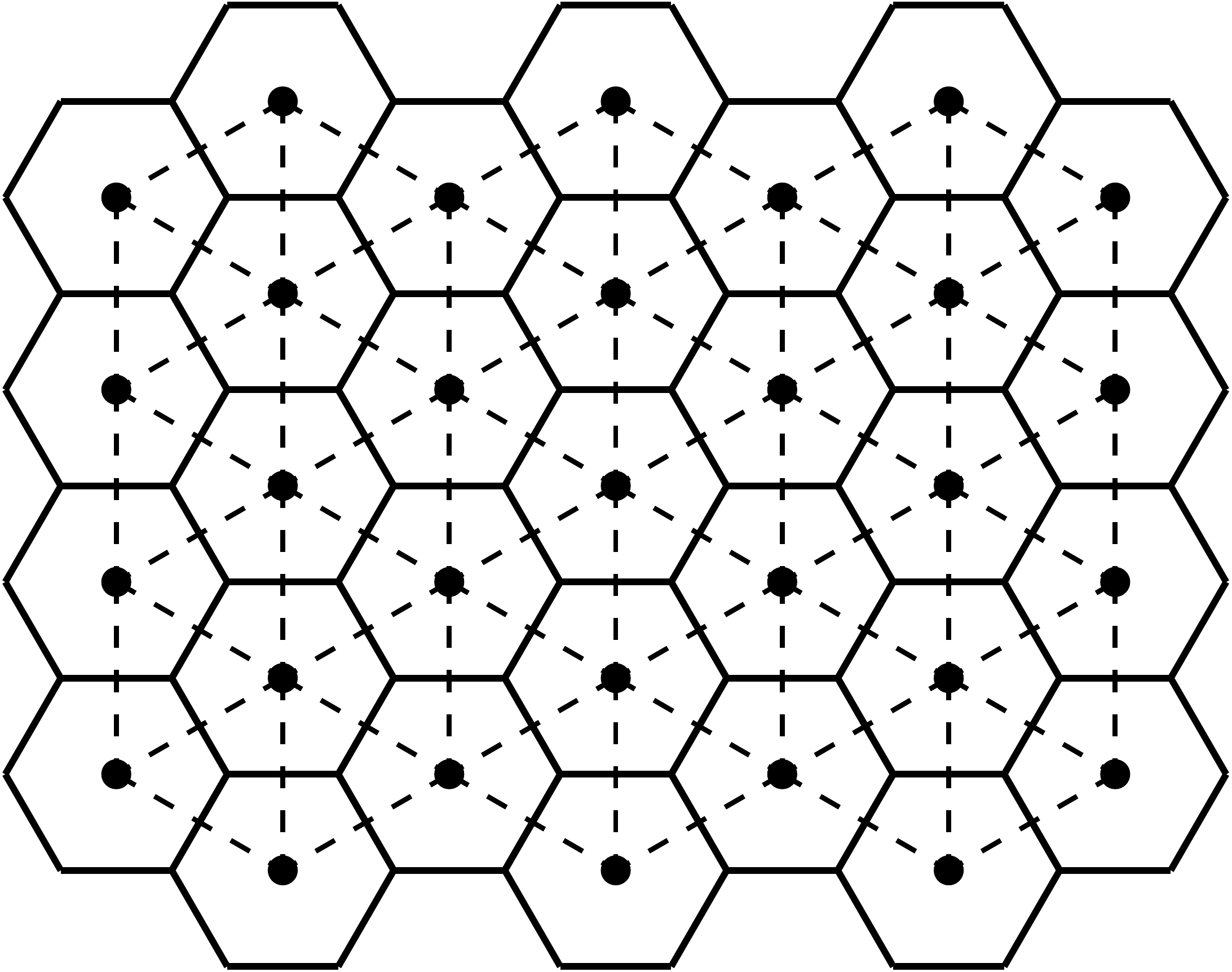}
\end{center}
\caption{The ideal sensor node layout which yields the target
  distance function, $H$.} \label{fig:model_uniform}
\end{figure}

The target distance function, $\cal{F}$, may depend on the
environment and on the application. In our preliminary work, the EvenCover
algorithm \cite{ChuSet2009}, we used a 2-dimensional plane with a
non-ideal target distance function derived from a Poisson point process. In
this work, we recognize that the Poisson distribution does not offer
an ideal target distance function for improving the quality of
representation. Given a 2-dimensional planar
region of interest, it is known that the best representation is
achieved when the region is tessellated in a space-filling manner by
hexagonal cells with a sensor node placed at the centroid of each
cell \cite{PacAga1995,BraMos2005}. This ideal layout is shown in
Fig.\,\,\ref{fig:model_uniform} and we denote by $H$ the target
distance function based on this layout. For a 3-dimensional region of
interest, the ideal layout will be different and likely based on one
of the space-filling tessellations of 3-dimensional space discussed
in \cite{AlaHaa2006}. We propose a new algorithm which accepts
any arbitrary target distance function, $\cal{F}$. Further, the algorithm
presented here, EvenRep($\cal{F}$,\,$L$), allows a limit, $L$, on the
number of active neighbors that a node will consider in its decision
making process. The choice of $L$ in our implementation of the
algorithm is discussed in Section~\ref{EvenCoverH}.

\subsection{Rationale and pseudo-code}
\label{evenCover}

\begin{figure}[!t]
\begin{center}
\framebox[10cm]{
\parbox{1.0cm}{
{\small
\begin{tabbing}
ww \= ww \= ww \= ww \= ww \= ww \= ww \= ww \= ww \kill
\mbox{~}\\
\mbox{~}01: \> {\em Initialization}:\\
\mbox{~}02: \> \> Turn on sense mode with probability $C_z$.\\~\\
\mbox{~}03: \> {\em EvenRep($\cal{F}$,\,$L$) Algorithm (executes in a loop)}:\\
\mbox{~}04: \> \> {\bf do}:\\
\mbox{~}05: \> \> \> {\bf if} node is active:\\
\mbox{~}06: \> \> \> \> $Q \leftarrow 1$\\
\mbox{~}07: \> \> \> {\bf else}:\\
\mbox{~}08: \> \> \> \> $Q \leftarrow 0$\\
\mbox{~}09: \> \> \> Wait for a random length of time between $0$ and $T$\\
\mbox{~}10: \> \> \> $K \leftarrow \min$( number of active neighbors,
$L$ )\\
\mbox{~}11: \> \> \> Compile list of $K$ nearest active neighbors\\
\mbox{~}12: \> \> \> {\bf for} $1 \leq k \leq K$:\\
\mbox{~}13: \> \> \> \> $X_k \leftarrow $ distance to $k$-th nearest active neighbor\\
\mbox{~}14: \> \> \> \> $Q \leftarrow Q + ~T_k( \cal{F} )$$/X_{k}$\\
\mbox{~}15: \> \> \> {\bf if} $( Q \geq z\pi X_{K}^2 )$:\\
\mbox{~}16: \> \> \> \> {\bf if} $( Q -z\pi X_{K}^2 \geq 0.5 )$:\\
\mbox{~}17: \> \> \> \> \> Set node to sleep mode\\
\mbox{~}18: \> \> \> \> {\bf else}:\\
\mbox{~}19: \> \> \> \> \> Set node to sleep mode with\\
\> \> \> \> \> \> probability $( Q - z\pi X_{K}^2 )$ \\
\mbox{~}20: \> \> \> {\bf else}:\\
\mbox{~}21: \> \> \> \> {\bf if} $( z\pi X_{K}^2 - Q \geq 0.5 )$:\\
\mbox{~}22: \> \> \> \> \> Set node to active mode\\
\mbox{~}23: \> \> \> \> {\bf else}:\\
\mbox{~}24: \> \> \> \> \> Set node to active mode with\\
\> \> \> \> \> \> probability $(z\pi X_{K}^2 - Q)$ \\
\mbox{~}25: \> \> {\bf while} true\\
\end{tabbing}
} } }
\end{center}
\caption{The EvenRep($\cal{F}$,\,$L$) algorithm executed at each
node.} \label{evenCoverCode}
\end{figure}

As before, let $N$ denote the total number of nodes in the region,
$A$ the area of the region and $z$ the desired spatial density.  Let
$C$ denote the {\em active node ratio}, the fraction of nodes in the
region of interest that are active. Note that $z$ is a property of
the application and not of the sensor network used by the
application, while $C$ describes the state of the sensor network.
Let $C_z = zA/N$ denote the desired active node ratio. Denote by
$T_k(\cal{F})$ the distance between the current node to its $k$-th
nearest active node in the target distance function
$\cal{F}$.

The pseudo-code for the algorithm executed at each node is shown in
Fig.\,\,\ref{evenCoverCode}. Each node can compute a quantity $Q$ at
the point where it is located based on a comparison between the
actual distances to its nearest neighbors and their target values.
Let $K$ denote the minimum of $L$ and the number of active neighbors within
the node's communication radius. Let $X_k$ denote the distance between
the node and its $k$-th nearest active neighbor. Then, in the case in
which the target distances are exactly achieved, $Q$ should equal the
actual number of active neighbors within a radius of $X_K$. The expected number of
nodes within radius $X_K$ at the desired active node ratio is $z\pi
X_K^2$. A node should stay in its current mode or switch to a
different mode
depending on whether or not the action taken helps bring the $Q$
computed by it closer to $z\pi
X_K^2$. For example, if the $Q$ computed is lower than that implied in the
active node ratio, the node should go into the sense mode if not
already in the sense mode.

Let $Q_i$ denote the $Q$ computed by
node $i$. The algorithm computes $Q_i$ as:
\begin{equation}
 Q_i = \delta_i + \sum_{k=1}^{K_i}
\frac{T_k(\cal{F})}{X_{k,i}}
\end{equation}
where $X_{k,i}$ is the distance from node $i$ to its $k$-th nearest
neighbor, $K_i$ is the minimum of $L$ and the number of active
neighbors of node $i$, and $\delta_i$ is given
by:
\begin{eqnarray}
\delta_i & = & \left\{
\begin{array}{ll}
1, & \mbox{if node}~i~\mbox{is in sense (active) mode},\\
0, & \mbox{if node}~i~\mbox{is in sleep mode.}
\end{array}
\right.
\end{eqnarray}
Let $r_i$ denote the distance $X_{K,i}$. If $Q_i$ computed as above
exceeds $z\pi r_i^2$ by $0.5$ or more and the node is in active mode,
turning it to the sleep mode will bring the local active node ratio
closer to that corresponding to the desired spatial density. Note that
when a node goes from active to sleep mode, the number of active nodes
in the local region of radius $r_i$ reduces by $1$ and, therefore,
comparing the difference between $Q_i$ and $z \pi r_i^2$ against
$0.5$ allows the node to best decide if it should sense or sleep so
that $Q_i$ is as close to $z \pi r_i^2$ as possible after it makes
the decision. Similarly, if $z\pi r_i^2$ exceeds $Q_i$ by $0.5$ or
more and node $i$ is in sleep mode, turning it to the active mode
will also bring the local active node ratio closer to that
corresponding to the desired spatial density. If $Q_i$ exceeds
$z\pi r_i^2$ by less than $0.5$, the algorithm sets the node to
sleep mode with probability equal to $Q_i - z\pi r_i^2$ (this does
not necessarily bring the local spatial density closer to the
desired value but is an attempt to bring the overall active node
ratio of the network closer to that corresponding to the desired
spatial density). Similarly, if $z\pi r_i^2$ exceeds $Q_i$ by less
than $0.5$, the node is set to active mode with probability $z\pi
r_i^2 - Q_i$.

\subsection{Complexity analysis}

We now examine the computational, communication and memory complexity
of the EvenRep($\cal{F}$,\,$L$) algorithm separated from any topology
control algorithm running at each node and entrusted with maintaining
a list of active neighbors within the node's communication
radius. Note that the execution of line 11 has computational
complexity $O(1)$ because each sensor node has to compile $K \leq L$
active neighbors where $L$ is a small constant. During the execution
of lines 12-14, the node updates $Q$ based on the distances between
itself and up to $K$ neighbors, each of which takes $O(1)$
time. Again, since $K \leq L$, a small constant, this too adds a
complexity of only $O(1)$. The rest of the algorithm is also $O(1)$
and therefore, the computational complexity of the algorithm at each
node is $O( 1 )$. The communication complexity of
EvenRep($\cal{F}$,\,$L$) is also $O(1)$, assuming again that a
separate and independent topology control algorithm maintains a list
of active neighbors. The memory usage of the algorithm is in the order
of $O(1)$. 

Note that the topology control algorithm, running independently of
EvenRep($\cal{F}$,\,$L$), may have its own communication complexity
associated with each node having to broadcast its active state to its
neighbors, receiving an acknowledgement for it from each of its active
neighbors and also receiving such state from each of its active
neighbors. This communication complexity depends on the topology
control algorithm used (which can vary greatly depending on the
algorithm \cite{San2006,LiHou2005-1313,SetGer2010}) and also on method used to estimate the
distance to each active neighbor (such as whether it is based on
assuming GPS devices in the sensor nodes \cite{NicNat2003} or on
exchanging transmission and reception powers \cite{KriBie1997}).

\subsection{EvenRep($H,L$)}
\label{EvenCoverH}

In the following, we focus on $\cal{F}$ $= H$ as the target distance
function and compute $T_k(H)$ for use in the EvenRep($H,L$)
algorithm. In the perfect layout (shown in
Fig.\,\,\ref{fig:model_uniform}), upon which the target distance function $H$
is based, each active node has six equidistant active
neighbors. Assuming an arbitrary distribution of sensor 
nodes in a region with a given node density $z$, the expected radius
of a circular region that contains seven active nodes is $r =
\sqrt{7/(\pi z)}$. To achieve a correspondence to a perfectly uniform
distribution such as in Fig.\,\,\ref{fig:model_uniform}, the heuristic
EvenRep($H,L$) uses a target distance of $\sqrt{7/(\pi z)}$ from an
active node to each of its six active neighbors.

Coincidentally and conveniently, for most topology control algorithms,
the number of active neighbors is typically $6$ or smaller
\cite{San2006,LiHou2005-1313,SetGer2010}. Therefore, we provide here the target
distance function only for $k \leq 6$:
\begin{equation}
T_k(H)=\sqrt{\frac{7}{\pi z}},\  \mbox{if}~ k \leq
6\\\label{uniform}
\end{equation}
The EvenRep($H,L$) algorithm analyzed in this report uses the above
expression for the target distance function $H$, and with $L$ set
equal to 3. This choice of $L$ is based on our
finding that the fourth neighbor and beyond (i) have a diminishing
influence on the quality of representation within the local region of
the node, and (ii) are better and more effectively considered at other
nodes for which they are one of the three closest neighbors.

\section{Simulation Results}
\label{results}

Our simulation experiments use $1000$ sensor nodes located in a square
region of unit area with a spatial distribution given by a Poisson
point process (each point in the region is equally likely to have a
node). In our implementation of EvenRep($H,3$), we choose $T$ as equal
to 10 units of time (recall from line 9 of the pseudo-code in
Fig.\,\,\ref{evenCoverCode} that each node waits a random length of
time between $0$ and $T$ between making the sense/sleep
decisions). The desired spatial density, the corresponding active node
ratio and the communication radius used in the experiments are
described as we discuss each of the simulation experiments in the
following subsections.

Each data point reported in the figures in this section represents
an average of $200$ different simulation experiments (each using
a different initial layout of the nodes). Based on the method of batch
means to estimate confidence intervals, we have determined that the
95\% confidence interval is within $\pm 1$\% for each of the data
points reported in the graphs.

\begin{figure*}[!t]
\begin{center}
\subfloat[{Average representation error, $D(G,R)$, achieved by
EvenRep($H,3$) plotted against time from $0$ to $10T$ for different
values of the communication radius.}]{
        \label{fig:task2_distance_samedensity}
        \includegraphics[width=3.25in]{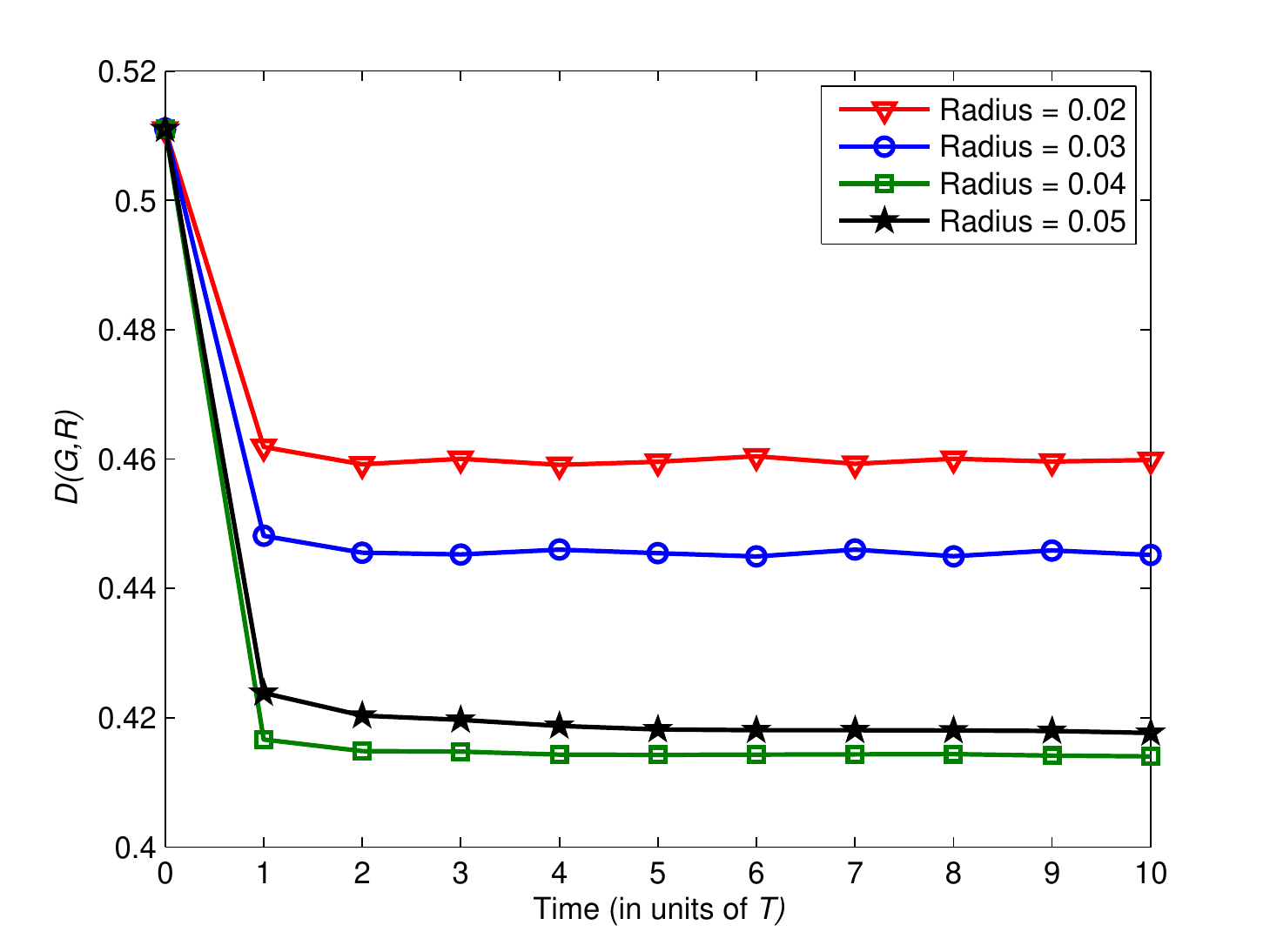}
        }
    \hskip0.2in 
\subfloat[{Unevenness of representation error, $U(G,R)$, achieved by
EvenRep($H,3$) plotted against time from $0$ to $10T$ for different
values of the communication radius.}]{
        \label{fig:task2_gini_samedensity}
        \includegraphics[width=3.25in]{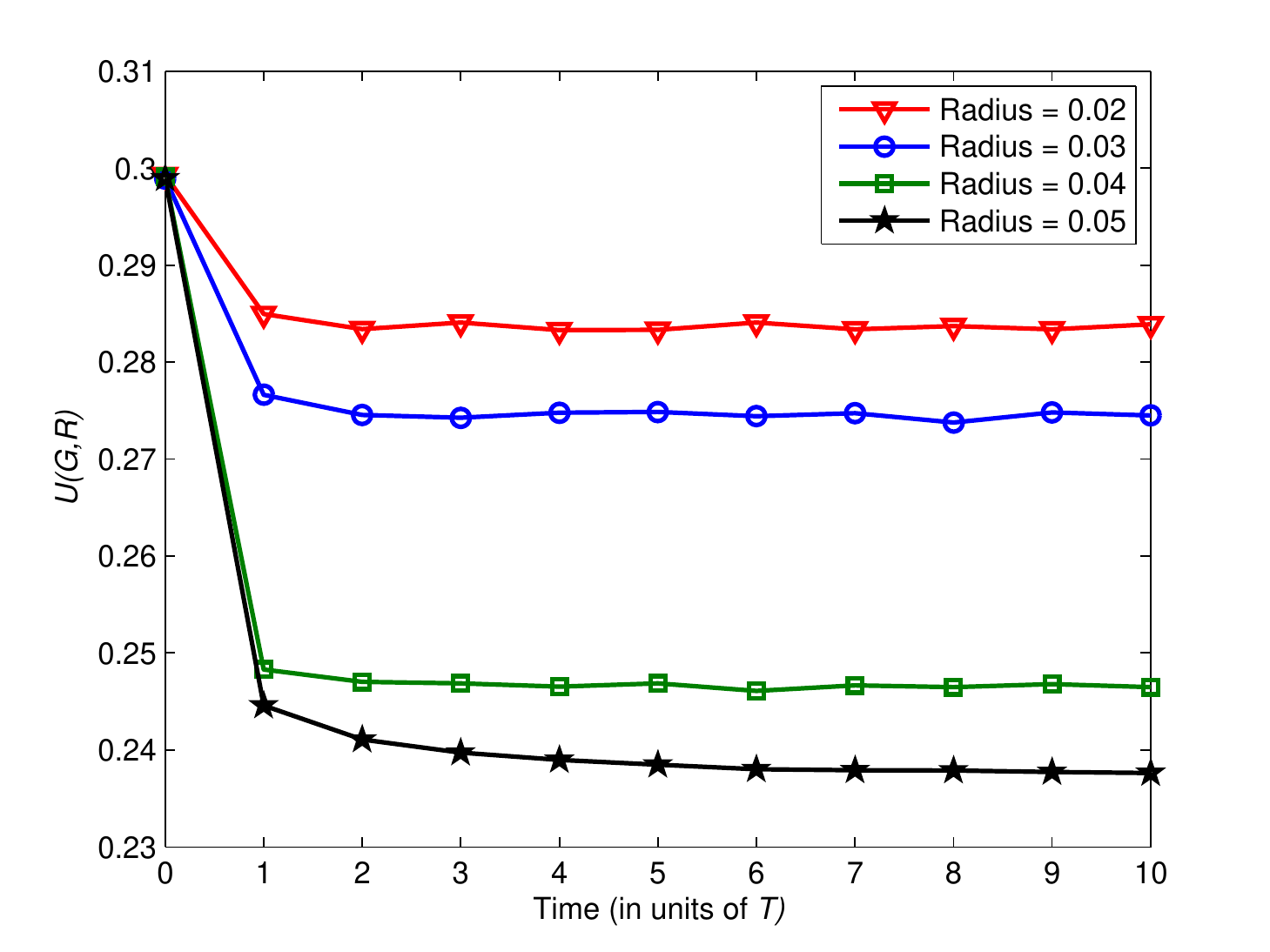}
        }
    \hskip0.2in 
\caption{Plots showing the convergence of EvenRep($H,3$) as the
algorithm executes and improves the quality of representation. In
these experiments, the desired spatial density used corresponds to
an active node ratio of $0.35$.}\label{fig:samedensity}
\end{center}
\end{figure*}

\begin{figure*}[!t]
\begin{center}
\subfloat[{Average representation error, $D(G,R)$, achieved by
EvenRep($H,3$) plotted against time from $0$ to $10T$ for different
values of the desired active node ratio. }]{
        \label{fig:task2_distance_sameradius}
        \includegraphics[width=3.25in]{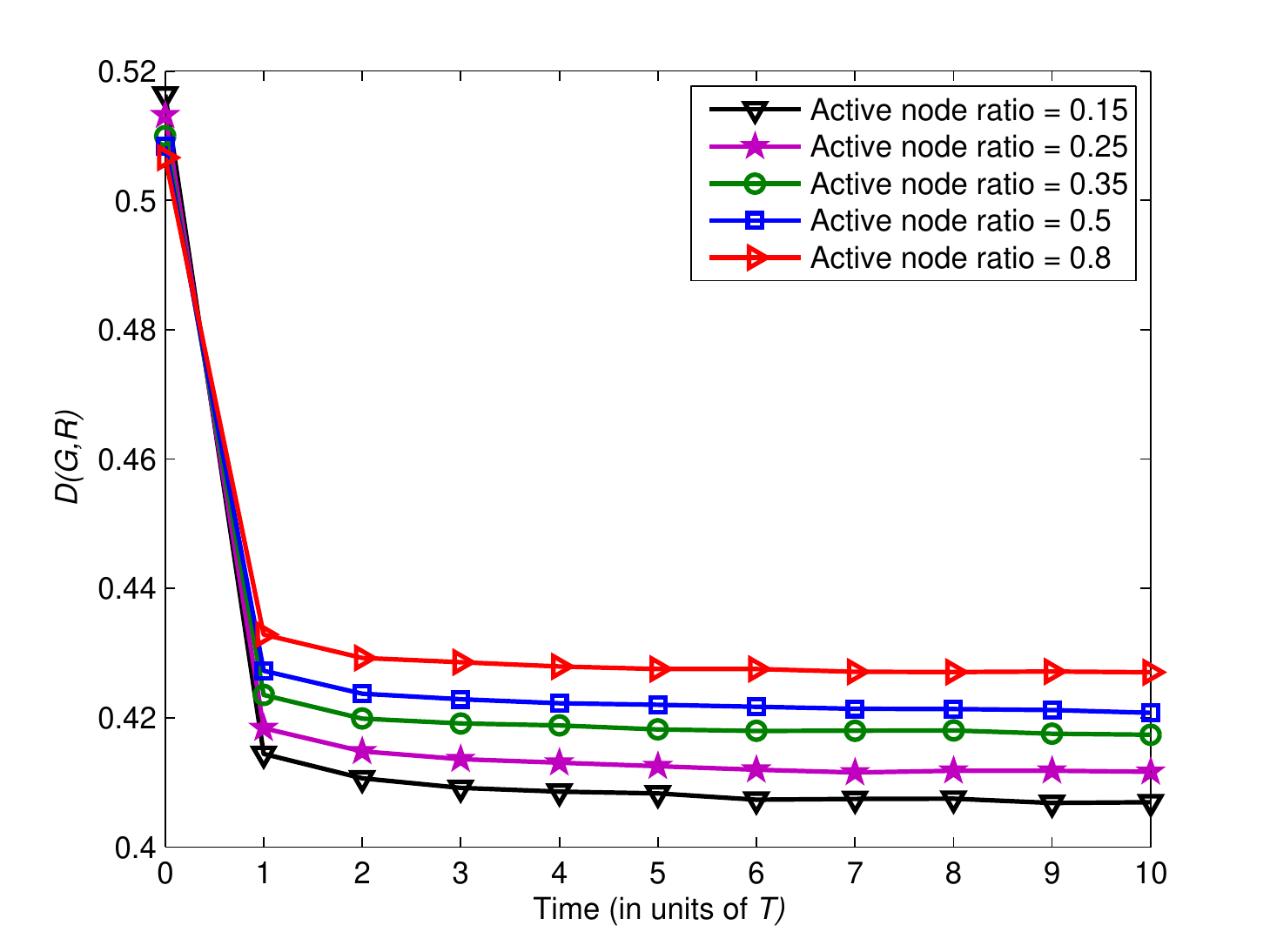}
        }
    \hskip0.2in
\subfloat[{Unevenness of representation error, $U(G,R)$, achieved by
EvenRep($H,3$) plotted against time from $0$ to $10T$ for different
values of the desired active node ratio.}]{
        \label{fig:task2_gini_sameradius}
        \includegraphics[width=3.25in]{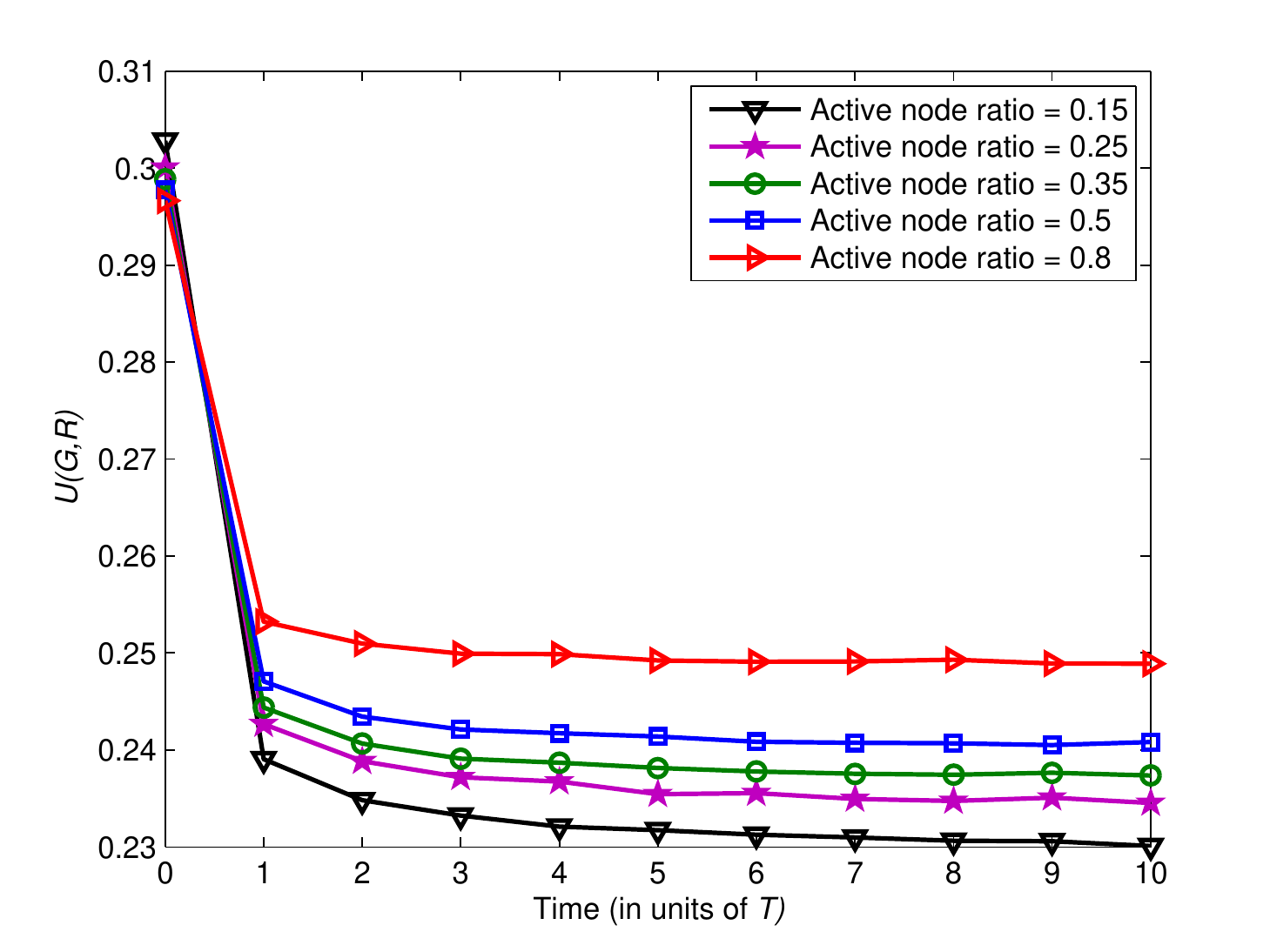}
        }
        \hskip0.2in
\caption{Plots showing the convergence of EvenRep($H,3$) as the
algorithm executes and improves the quality of representation. In
these experiments, the communication radius used is
0.08.}\label{fig:sameradius}
\end{center}
\end{figure*}

Our algorithm begins with each node randomly setting itself to
active mode with probability equal to $C_z$, the expected active
node ratio when the desired spatial density is $z$. Thus, the
spatial distribution of active nodes at the beginning of the
simulation is given by a finite Poisson point process. To understand
the reference point at which the simulation begins, we prove an
additional set of results in Appendix~\ref{sec:poisson} on the two
metrics. Theorem \ref{theorm:poisson_metric} in
Appendix~\ref{sec:poisson} proves that when the active nodes are
located in the region with a spatial distribution given by a Poisson
point process, the expected value of the average 
representation error, $D(G,R)$, is $0.5$ and the expected
value of the unevenness of representation error, $U(G,R)$, is
$1-1/\sqrt{2} \approx 0.293$. Due to the use of a {\em finite} Poisson point
process for the initial layout of the nodes in the unit area in our
simulation experiments, border
effects cause the initial values of $D(G,R)$ and $U(G,R)$ to be
slightly larger than $0.5$ and $1-1/\sqrt{2}$, respectively. 
Recall from Theorem
\ref{theorem:lowerBound} that a lower bound on $D(G,R)$ is
$2/(3\sqrt{\pi}) \approx 0.376$ and a lower bound on $U(G,R)$ is
$0.2$. Therefore, in our simulation experiments, one should expect
that $D(G,R)$ reduces to something between $0.5$ and $0.376$ and
that $U(G,R)$ reduces to something between $0.293$ and $0.2$ after a
certain length of time since the beginning of algorithm execution.

\subsection{Convergence}

\begin{figure*}[!t]
\begin{center}
\subfloat[{Average representation error, $D(G,R)$, achieved after
time $5T$ by different algorithms plotted against the communication radius.}]{
        \label{fig:task1_distance}
        \includegraphics[width=3.25in]{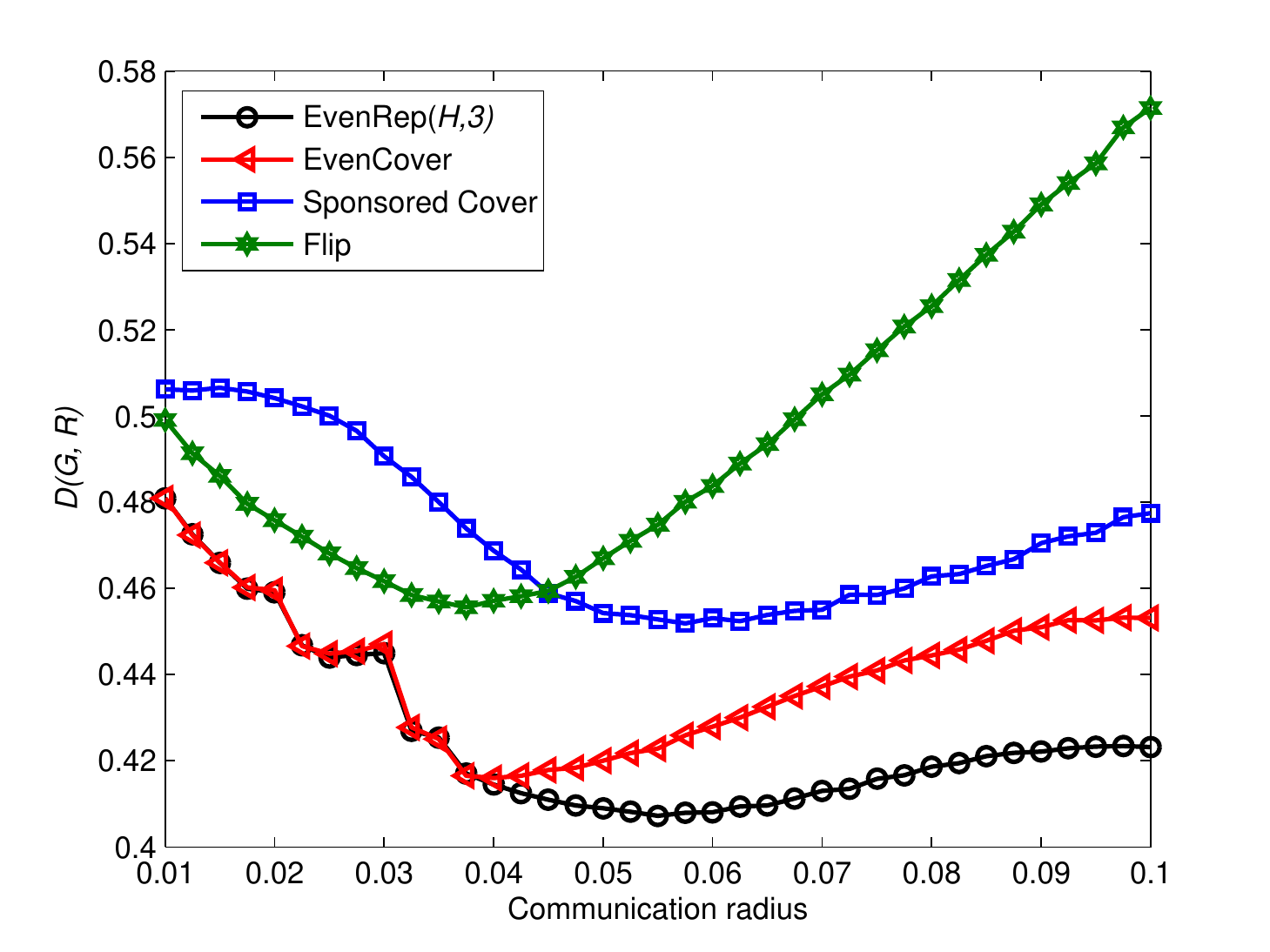}
        }
        \hskip0.2in
\subfloat[{Unevenness of representation error, $U(G,R)$, achieved
after time $5T$ by different algorithms plotted against the communication radius.}]{
        \label{fig:task1_gini}
        \includegraphics[width=3.25in]{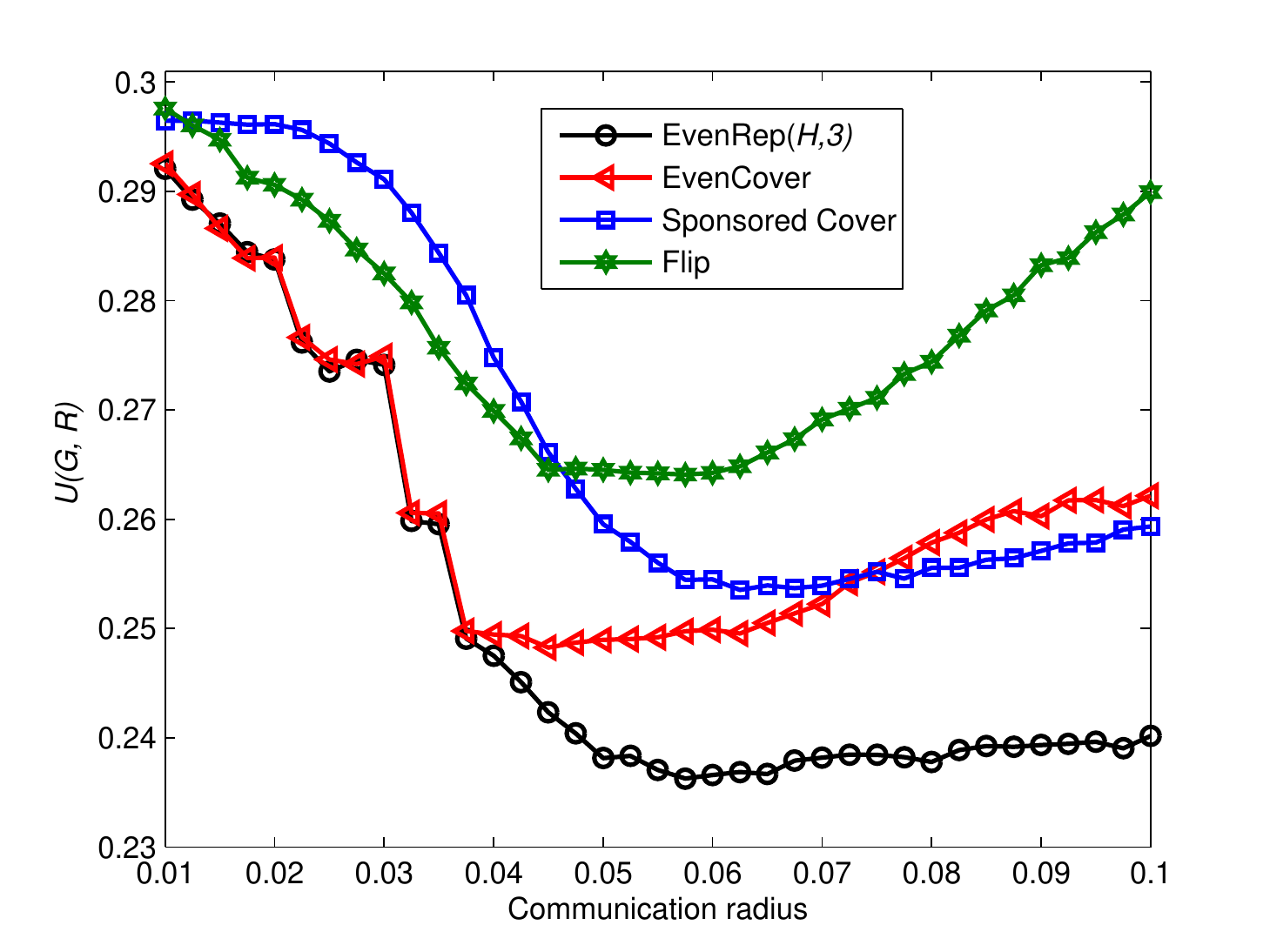}
        }
    \hskip0.2in 
\caption{Plots comparing the quality of representation achieved by
different algorithms.}
\end{center}
\end{figure*}

Fig.\,\,\ref{fig:samedensity} shows the convergence of the
EvenRep($H,3$) algorithm on the quality of representation metrics,
$D(G,R)$ and $U(G,R)$, for different values of the communication
radius while the desired spatial density corresponds to an active
node ratio equal to $0.35$ (i.e., given $1000$ nodes in the unit
area in our simulations, the desired spatial density, $z$,
corresponds to $350$ active nodes). Almost all topology control
algorithms achieve an average communication radius at each node
corresponding to six or fewer neighbors
\cite{LiHou2005-1313,SetGer2010}. Therefore, the largest communication
radius we use is $0.08$ units, which corresponds to approximately $6$
active neighbors within a node's communication radius.

Fig.\,\,\ref{fig:task2_distance_samedensity} plots the average
representation error as the algorithm continues to execute for a
length of time equal to $10T$.
Fig.\,\,\ref{fig:task2_gini_samedensity} shows the corresponding
convergence of the EvenRep($H,3$) algorithm on the unevenness of
representation error metric using the same set of parameters. Note that both the average and the
unevenness of representation error reduce rapidly as early as $T$
(about $2$ iterations of the loop between lines 04--25 in
Fig.\,\,\ref{evenCoverCode} because the expected length of time
between two iterations is $T/2$). It is not necessarily true that as
the communication radius increases, the average representation error,
$D(G,R)$, will reduce. This is because, when the communication radius is small,
in order to achieve the desired active node density more nodes than
necessary will determine that they should be active since they have limited
information about the status of other nodes in the region. This
reduces the average distance between active nodes, thus
reducing $D(G,R)$, but does not necessarily improve $U(G,R)$. As one
might expect, as the communication radius increases, each sensor
node has more neighbors and is able to collect significantly more
relevant information about the quality of representation in its
neighborhood, thus reducing the unevenness of representation error.

In our second set of experiments on the convergence properties of
EvenRep($H,3$), we keep the communication radius constant and vary
the desired active node ratio from 0.15 to 0.8. We use a communication
radius of 0.08 units corresponding to an average of about six active
neighbors within the radius when the active node ratio is 0.35 (the
density used in the previous set of experiments). The result is
plotted in Fig.\,\,\ref{fig:task2_distance_sameradius} and
Fig.\,\,\ref{fig:task2_gini_sameradius} for an interval of time up
to $10T$. Once again, the algorithm appears to converge rapidly
within time $T$, which corresponds to approximately $2$ executions
of the algorithm for each sensor node. Note that the algorithm's
performance does not increase with the increase of the active node
ratio. In this set of experiments, the algorithm achieves its best
performance when the active node ratio is 0.15, the smallest ratio
used in the experiments. This is expected because, when the active node
ratio is high, the algorithm has fewer choices in determining the
active node layout and becomes more confined to the original Poisson
distribution of the nodes; on the other hand, when the active node
ratio is low, the algorithm has more choices in determining which
nodes should sense and which should sleep, leading to an improved
quality of representation. A theoretical proof of the convergence of
the algorithm remains an open problem.

It should be noted that, while fast convergence to low values of these
metrics is desirable, convergence to one particular layout of the
active nodes is not desirable. This is because the lifetime of a
network suffers if nodes, once chosen to be active, remain active
forever. 

\subsection{Comparative analysis}
\label{sec:comparison}

\begin{figure*}[!t]
\begin{center}
\subfloat[{The fraction of nodes alive as the execution of
an algorithm proceeds.}]{
        \label{fig:task3_lifetime}
        \includegraphics[width=3.25in]{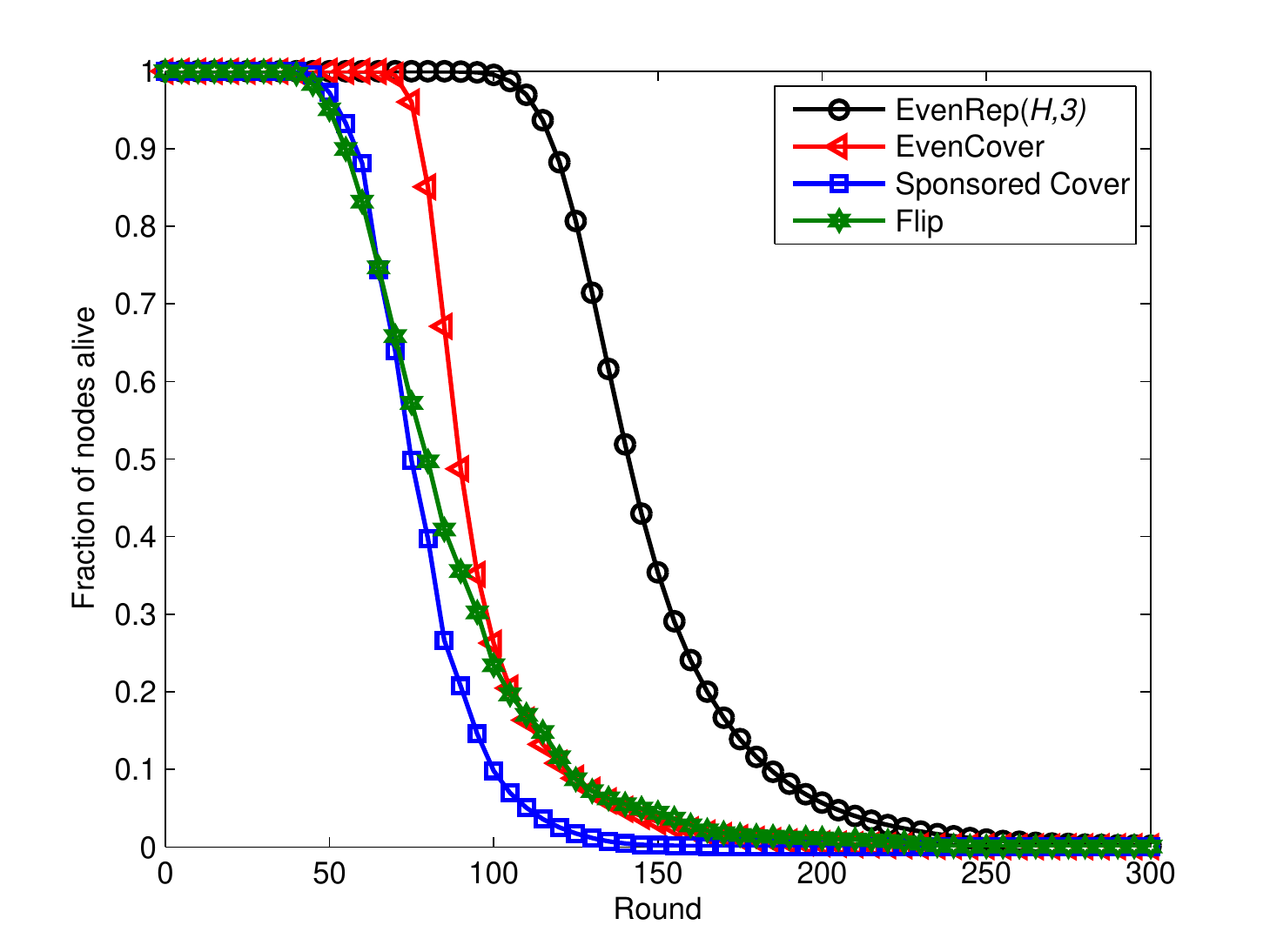}
        }
    \hskip0.2in 
\subfloat[{The quality of representation as the execution of an
algorithm proceeds.}]{
        \label{fig:task3_gini}
        \includegraphics[width=3.25in]{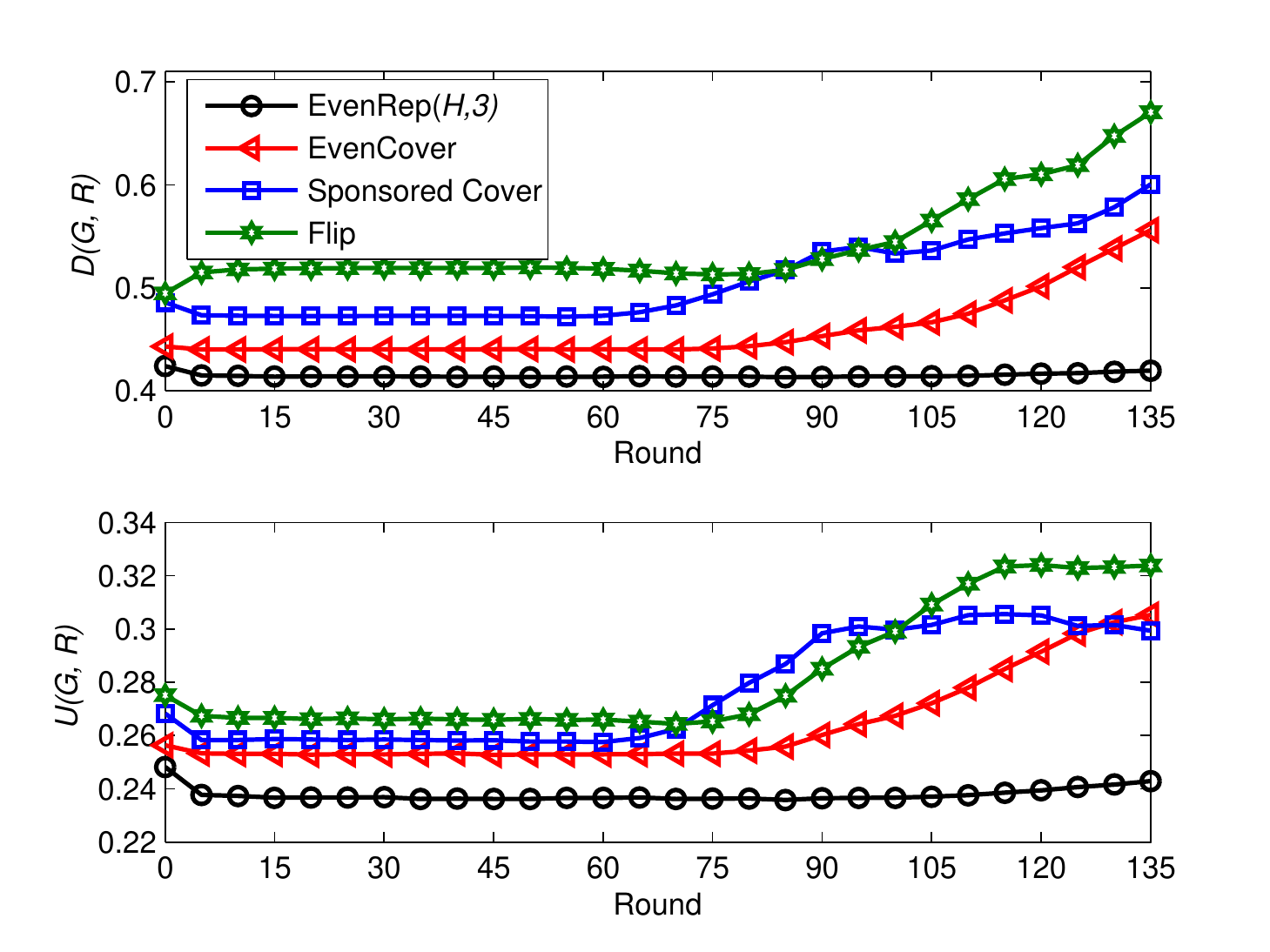}
        }
        \hskip0.2in
\caption{Plots showing the network lifetime (the fraction of nodes
  alive is used as in indication of network lifetime) and the
  degradation in the quality of representation as nodes die. The
  communication radius used is $0.08$ units.}
\end{center}
\end{figure*}

We report results for the following distributed algorithms:
\begin{itemize}
\item The EvenRep($H,3$) algorithm: The EvenRep($\cal{F}$,\,$L$) algorithm
  using a target distance function $\cal{F}$ $= H$ (based on a
  node placement in which the region is covered by non-overlapping
  hexagonal cells with each sensor node covering one cell) and
  $L=3$. The value of $T_k(H)$ used in the algorithm implementation is
  discussed in Section~\ref{EvenCoverH} (see Eqn.\,\,(\ref{uniform})).
\item The EvenCover algorithm: This is our preliminary work \cite{ChuSet2009} that
employs a target distance function derived from a distribution of
sensor nodes given by a Poisson point process.
\item Sponsored Cover: For a representative coverage protocol that
  assumes an area-sensing application, we choose the well-cited
  coverage-preserving node-scheduling scheme based on sponsored
  coverage calculations \cite{TiaGeo2002}. In this protocol,
  a node decides to go into the sleep mode if its entire designated
  sensing area is also covered by its neighbors. To avoid situations
  in which each of two neighbors expects a certain spot to be covered
  by the other, the protocol implements a random time for which each
  node delays its decision.
\item Flip: In this protocol, each node counts the fraction of its
  neighbors (including itself) that are active and sets itself into
  either the sleep mode or the active mode depending on whether or not
  this fraction is larger or smaller than the desired active node
  ratio. We use this algorithm as representative of coverage strategies
  based on distributed algorithms for graph coloring that do not use
  the distance between the nodes in their computations \cite{OdoSet2004}.
\end{itemize}

Figs.\,\,\ref{fig:task1_distance} and \ref{fig:task1_gini} plot the performance of these
algorithms against the communication radius after the algorithms execute for a
period of time $5T$. The Sponsored Cover algorithm takes the
communication radius as the sole input while the other three also
require the desired active node ratio as an input (chosen as $0.35$ in
these figures). The figures
show that EvenRep($H,3$) achieves the best quality of
representation among the algorithms studied. Given that the target
distance function $H$ is based on a more spatially uniform
distribution of nodes than a Poisson point process (upon which the
EvenCover algorithm is based), one expects the EvenRep($H,3$)
algorithm to perform better than EvenCover (as is also observed in
these results). The Sponsored Cover algorithm seeks full, but not
necessarily even, coverage of points in the region. As a result, in
parts of the region with a denser cluster of nodes, the Sponsored
Cover algorithm will unnecessarily turn on larger numbers
of nodes resulting in poorer spatial uniformity. The Flip algorithm,
on the other hand, does not consider distances between nodes and
therefore, is far from being able to achieve a good quality of
representation.

\subsection{Network Lifetime}

In this section, we compare the network lifetime of EvenRep($H,3$)
with other algorithms. In our simulation experiments, we begin
with each node allocated a certain amount of energy which is
expended in the following two ways:
\begin{itemize}
\item On/off broadcast transmission, used to broadcast the node's
new status (on/off) to its neighbors.
\item Reception, for receiving data and control information from
  neighbors.
\end{itemize}
The power
consumption model is adapted from \cite{TiaGeo2002}, with the
assumption that each node is allocated 0.05J of energy and the data
signal is a 2000-bit report message. The transmission energy
consumption and the reception energy consumption is calculated as follows:
\begin{eqnarray}
E_{T_x}(d) & = & E_{elec}\times k+\varepsilon_{friss\mbox{-}amp}\times k\times d^2 \\
E_{R_x} & =  &E_{elec}\times k
\end{eqnarray}
where $E_{T_x}(d)$ is the energy consumed in transmitting the signal
to an area of radius $d$, $E_{elec}$ is the energy consumed for the
radio electronics, $\varepsilon_{friss\mbox{-}amp}$ is for the power
amplifier and $E_{R_x}$ is the energy consumed in receiving the
signal. Radio parameters are set as $E_{elec}=50 nJ/$bit and
$\varepsilon_{friss\mbox{-}amp}=10 pJ/$bit$/m^2$. The energy
consumed in sensing is not considered in this set of simulations,
exactly as in the model used in \cite{TiaGeo2002}.

The EvenRep($\cal{F}$,\,$L$) algorithm does not use the concept of
{\em rounds} while Sponsored Cover and Flip do. In
order to present a fair comparison, therefore, we use a
slightly modified version of our algorithm so that each node is
scheduled to work in rounds. At the beginning of each round, each
node randomly picks a start time, $t$, between $0$ and $T$, and
makes its decision based on the information received so far. At time
$t$ after the start of the round, it switches its state according to
the decision made. If the node chooses to turn on, then it will
broadcast its decision to its neighbors. Each of its neighbors will
receive the decision and remember it for its own reference. If the
node chooses to turn off, it will not make any broadcast
attempts. A node is considered dead if it has consumed all the power
allocated and alive, otherwise. Once a node is dead, it cannot be
turned on again nor can it broadcast or receive signals. 

A fair lifetime comparison can only be achieved when both the communication radius and the active
node ratio are the same for each of the four listed algorithms. Recall
that the Sponsored Cover algorithm does not require the active node
ratio as an input. However, since the active node ratio achieved by
the Sponsored Cover algorithm is a function of the communication
radius employed, it strictly limits the choice of the desired active
node ratio we can use in our comparisons. In these experiments, we use
a communication radius of $0.08$ units for all four algorithms. Using
this communication radius, we first note the active node ratio
achieved by the Sponsored Cover algorithm and then use this ratio as
the input active node ratio for the other three algorithms.

Fig.\,\,\ref{fig:task3_lifetime} reports the fraction of nodes alive
as time progresses to indicate the network lifetime (for example,
one may define network lifetime as the time until 50\% of the nodes
are dead). Figs.\,\,\ref{fig:task3_lifetime} and
\ref{fig:task3_gini} demonstrate that the EvenRep($H,3$) algorithm
significantly improves the lifetime of the network while also
achieving a better quality of representation even as nodes die. Since the Sponsored Cover
algorithm seeks area-sensing coverage of points in the region, some
sensor nodes may have to constantly stay active while its neighbors
are constantly in sleep mode. The goal of full coverage, therefore,
contributes to a reduced lifetime while the goal of spatial
uniformity, appropriate for spot-sensing applications, results in an
improved lifetime. The combination of striving for spatial
uniformity and the use of random chance to place nodes in specific
modes (as in lines 19 and 24 in Fig.\,\,\ref{evenCoverCode}) leads
to the difference between the lifetime achieved by EvenRep($H,3$)
and that achieved by other algorithms. In addition to improved
lifetime, Fig.\,\,\ref{fig:task3_gini} shows that the EvenRep($H,3$)
algorithm better preserves the average and the evenness of
representation error as nodes in the network die, in comparison to
other algorithms. This indicates that EvenRep($H,3$) achieves a more
graceful degradation of the sensor network as the battery power in
the nodes are exhausted. The EvenCover algorithm, on the other hand,
does not achieve as good a lifetime. As can be observed from
Figs.\,\,\ref{fig:task1_distance} and \ref{fig:task1_gini}, the
EvenCover algorithm performs as well as EvenRep($H,3$) when the
communication radius is small (corresponding to 4 or fewer active
neighbors). However, a communication radius corresponding to an
average of 6 active neighbors is a more realistic scenario generated
by topology control algorithms and EvenRep($H,3$) performs
significantly better in this range of the communication radius.

\section{Conclusion}\label{conclusion}
Recent research literature has largely focused on area-sensing
applications where the goal is to get each point in the region
$k$-covered. In this report, we turn our attention to spot-sensing
applications and introduce a new problem with the goal of achieving
a good quality of representation by activating the sensor nodes in
such a way that active nodes are spatially uniformly distributed. To the best of the knowledge of the authors, this is the
first work that specifically targets quality of representation of
the points in the region for spot-sensing applications in sensor
networks. A
better quality of representation indicates a shorter normalized
average distance between the points in the region of interest to
their nearest active nodes, and an even distribution of these
distances. We have developed two complementary metrics to capture the
quality of representation achieved by a sensor node deployment and
used these in our evaluations of different algorithms.

We have developed a generalized distributed
algorithm called EvenRep($\cal{F}$,\,$L$), which accepts two
parameters: $\cal{F}$, a target distance function that maps $k \geq 1$
to the desired distance from a node to the $k$-th nearest active
neighbor and $L$, the maximum number of active neighbors that a node
will consider in its decision making. We implement an instance of this
algorithm using a target distance function, $H$, based on a spatial
distribution in which the 2-dimensional region is covered by
non-overlapping hexagonal cells with a sensor at the centroid of each
cell. The results show that EvenRep($H,3$) achieves a better quality
of representation and, very importantly, a longer network lifetime
than other related distributed algorithms.

Algorithms designed for area-sensing applications use the sensing
radius as the only input parameter to determine the sense/sleep
status of nodes. A given sensing radius implies a specific target
spatial density, and vice versa. If the target density is low, it
implies a large sensing radius and therefore, for most coverage
algorithms, a large communication radius and high energy costs.
Coverage algorithms designed for area-sensing applications,
therefore, cannot be adapted for spot-sensing applications,
especially at lower values of the desired density of active nodes.
The EvenRep($H,3$) algorithm, however, works well for spot-sensing
applications at all active node densities while also achieving a
longer lifetime. As a result, the EvenRep($H,3$) algorithm also
allows a graceful degradation of the network as nodes die because it
preserves the quality of representation at all active node
densities. Admittedly, EvenRep($H,3$) is a heuristic. Future work in
this direction should seek a strong theoretical foundation upon which
distributed algorithms for improved quality of representation may be based.

\appendices

\section{}\label{sec:uniform_bounds}

Given a region of interest, $R$, of arbitrary shape covered by a graph
of active sensor nodes, $G$, let $\mathrm{LB}(D)$ denote the lower
bound on $D(G,R)$ and let $\mathrm{LB}(U)$ denote the lower bound on
$U(G,R)$. The following theorem derives an upper bound on these lower
bounds.

\begin{theorem}
Given a region of interest of arbitrary shape covered by active sensor nodes,
\begin{eqnarray}
\mathrm{LB}(D) & \leq & \left(\frac{1}{9}+\frac{\ln 3}{12}\right)\sqrt{2\sqrt{3}}
\approx 0.3772 \\
\mathrm{LB}(U) & \leq & 0.2038
\end{eqnarray}
\label{theorm:achievablegini}
\end{theorem}
\begin{proof}
\begin{figure}[!t]
\begin{center}
\subfloat[{A closer look at the neighborhood of a sensor node in a
spatial distribution upon which the target distance function $H$ is
based.}]{
        \label{fig:node}
        \includegraphics[width=1.4in]{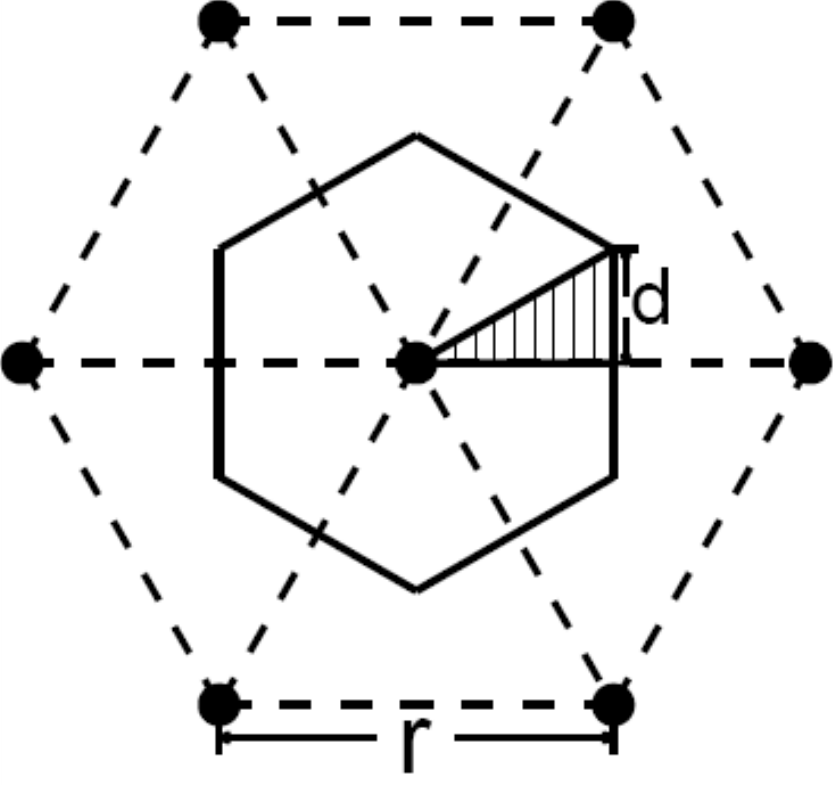}
        }
    \hskip0.2in 
\subfloat[{A closer look at the shaded triangle in (a).}]{
        \label{fig:triangle}
        \includegraphics[width=1.4in]{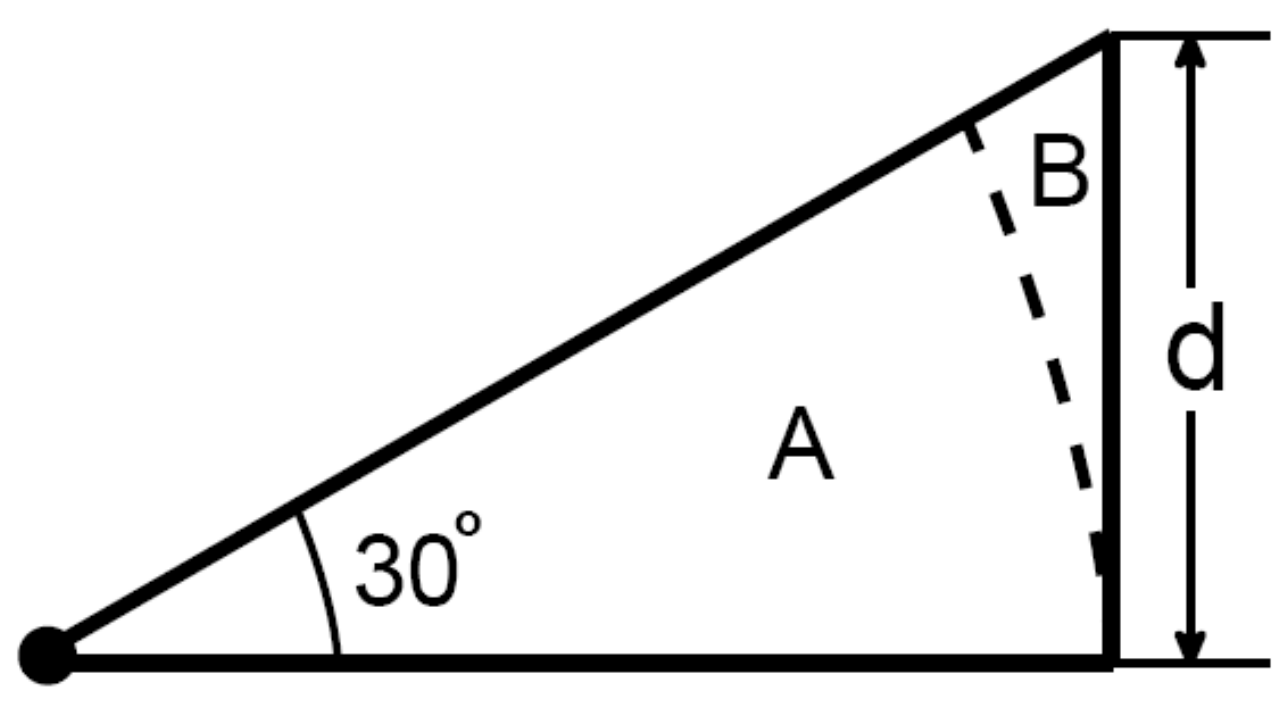}
        }
        \hskip0.2in
\caption{Geometric illustrations supporting the proof of
  Theorem~\ref{theorm:achievablegini}.}
\end{center}
\end{figure}
Given a 2-dimensional region of arbitrary shape, we know that
perfectly circular regions covered by each sensor node (as assumed
in the proof of Theorem\,\ref{theorem:lowerBound}) will not achieve
a space-filling tessellation of the area. Thus, the lower bounds on
the quality of representation metrics will be higher for regions of
arbitrary shape than those in Theorem\,\,\ref{theorem:lowerBound}. A
2-dimensional plane can achieve a regular symmetric tessellation
with only three types of tiles: equilateral triangles, squares, or
hexagons \cite{KinMoo2001}. When an active sensor node is placed at
the centroid of each tile, hexagonal tiles, being closer to a
circular shape, achieve better quality of representation than
triangles or squares. A region of interest of any arbitrary shape
can achieve a space-filling tessellation with an infinite number of
hexagonal tiles, each of infinitesimal size. Therefore, the quality
of representation achieved by infinite space-filling hexagonal
tiles, as shown in Fig.\,\,\ref{fig:model_uniform}, is the best
coverage that can be guaranteed for regions with an unknown
arbitrary shape. The metrics $D(G,R)$ and $U(G,R)$ for this scenario
represents, respectively, the upper bounds on the lower bounds of
$D(G,R)$ and $U(G,R)$ for regions of arbitrary shape.

Let $r$ denote the distance between each pair of neighboring active
sensor nodes (placed at the centroid of each hexagonal tile).
Consider each of the innermost triangles (made up of dashed lines in
Fig.\,\,\ref{fig:model_uniform}) with sensor nodes as vertices.
Consider the nodes near the center of the region (i.e., not at the
boundary); each node belongs to six triangles and therefore, each
triangle can be said to hold $3\times(1/6)=0.5$ nodes. Thus, given
$z$ nodes per unit area, the number of these triangles that can
cover a unit area is $z/0.5=2z$. Note that the area of each triangle
is $\sqrt{3}r^2/4$. Since the area covered by $2z$ of these
triangles is 1, we have $2z(\sqrt{3}r^2/4)=1$. Thus,
\begin{equation}
r=\sqrt{\frac{2}{\sqrt{3}z}}\label{equ:r}
\end{equation}For any
point within each hexagonal tile, the closest sensor node is the one
at the centroid of the tile. Since all the hexagonal tiles are
congruent and identical with respect to the sensor node within them,
$D(G,R)$ and $U(G,R)$ for the region of interest is the same as the
$D(G,R)$ and $U(G,R)$ for any one hexagonal tile. Consider one such
hexagonal tile, shown by solid lines in Fig.\,\,\ref{fig:node}. The
hexagonal tile can be divided into twelve non-overlapping congruent
right-angled triangles, one of which is shown shaded in
Fig.\,\,\ref{fig:node}. Since the triangles are all congruent and
also identical with respect to the placement of the nearest sensor
node, $D(G,R)$ and $U(G,R)$ for the hexagonal cell is the same as
the $D(G,R)$ and $U(G,R)$ for each triangle. Consider one such
triangle, shown in Fig.\,\,\ref{fig:triangle}.

Denote by $d$ the length of the shortest edge of the triangle in
Fig.\,\,\ref{fig:triangle}. Using elementary geometry,
$d=\sqrt{3}r/6$. The triangular area can be divided into two parts:
\begin{itemize}
\item Area A: the area in which the distance between any point in the
  region to the sensor node is no larger than $\sqrt{3}d$.
\item Area B: the area in which the distance between any point in the
  region to the sensor node is larger than $\sqrt{3}d$.
\end{itemize}
The two parts of the triangular region are shown in
Fig.\,\,\ref{fig:triangle}. Thus, the expected distance between a
point in the triangular area to the sensor node can be computed based
on the expected distances from points within each of these two parts:
\begin{eqnarray}
\overline{d}&=&\frac{\pi}{6}\int_0^{\sqrt{3}d} \frac{x^2
dx}{\frac{\sqrt{3}}{2}d^2}+\int_0^{\frac{\pi}{6}}\int_{\sqrt{3}d}^{\frac{\sqrt{3}d}{sin(\theta+\frac{\pi}{3})}}
\frac{x^2 dx d\theta }{\frac{\sqrt{3}}{2}d^2}\nonumber
\\&=&\left(\frac{2}{3}+\frac{1}{2}\ln 3\right)d \label{dbar}
\end{eqnarray}
Given $d=\sqrt{3}r/6$ and Eqns.\,\,(\ref{equ:strength of coverage})
and \,(\ref{equ:r}), we have:
\begin{eqnarray}
D(G,R)&=&\overline{d}\sqrt{z}=\left(\frac{2}{3}+\frac{1}{2}\ln 3\right)d
\sqrt{z} \nonumber\\
&=&\left(\frac{1}{9}+\frac{1}{12}\ln
  3\right)\sqrt{2\sqrt{3}} \nonumber \\
& \approx & 0.3772 \label{Dlimit}
\end{eqnarray}
We now proceed to derive $U(G,R)$. Consider two random
points whose distances to the sensor node are $x$ and $y$. There
are three cases where the locations of the two points may fall.
\begin{itemize}
\item Case 1: Both points fall within Area A.
\item Case 2: One of the points falls within Area A and the other within Area B.
\item Case 3: Both points fall within Area B.
\end{itemize}
We now consider each of the three cases. In the following,
$E_w[|x-y|]_i$ denotes the expected value of $|x-y|$ weighted by
the probability of Case $i$.

Case 1:
\begin{eqnarray}
E_w[|x-y|]_1&=&\int_0^{\sqrt{3}d}\int_0^y(y-x)\frac{\frac{\pi}{6} x
dx}{\frac{\sqrt{3}}{2}d^2} \frac{\frac{\pi}{6}y dy}{\frac{\sqrt{3}}{2}d^2}\nonumber \\
&   +&\int_0^{\sqrt{3}d}\int_0^x(x-y)\frac{\frac{\pi}{6}y
dy}{\frac{\sqrt{3}}{2}d^2}
\frac{\frac{\pi}{6}x dx}{\frac{\sqrt{3}}{2}d^2} \nonumber \\
&=&\frac{\sqrt{3}{\pi}^2 d}{45} \label{diffbar}
\end{eqnarray}

Case 2:

The probability that a point falls in Area A is
\[\int_0^{\sqrt{3}d} \frac{\frac{\pi}{6} x
dx}{\frac{\sqrt{3}}{2}d^2}\] and the probability that a point falls
in area B is:
\[\int_0^{\frac{\pi}{6}}\int_{\sqrt{3}d}^{\frac{\sqrt{3}d}{\sin(\theta+\frac{\pi}{3})}}\frac{y
dy d\theta}{\frac{\sqrt{3}}{2}d^2}\]
Define $k_\theta$ as:
\[k_\theta={\frac{\sqrt{3}d}{\sin(\theta+\frac{\pi}{3})}}\]
 Therefore,
\begin{eqnarray}
E_w[|x-y|]_2&=&2\int_0^{\frac{\pi}{6}}\int_{\sqrt{3}d}^{k_\theta}\int_0^{\sqrt{3}d}
(y-x)\frac{{\frac{\pi}{6} x dx}}{{\frac{\sqrt{3}}{2}d^2}} \frac{y dy
d\theta}{\frac{\sqrt{3}}{2}d^2}\nonumber \\
&=&\pi d(-\frac{2}{3}+\frac{2}{9}\sqrt{3}+\frac{1}{6}\sqrt{3}\ln 3)
\label{E_2}~~~~~~
\end{eqnarray}

Case 3:

The probability that both points fall in Area B is:
\[ \left( \int_0^{\frac{\pi}{6}}\int_{\sqrt{3}d}^{\frac{\sqrt{3}d}{\sin(\theta+\frac{\pi}{3})}}\frac{y
dy d\theta}{\frac{\sqrt{3}}{2}d^2} \right)^2\]
Define $k_\eta$ as:
\[k_\eta=\frac{\sqrt{3}d}{\sin(\eta+\frac{\pi}{3})}\]
Therefore,
\begin{eqnarray}
E_w[|x-y|]_3= \int_0^{\frac{\pi}{6}}
\int_{\sqrt{3}d}^{k_\theta}\int_0^{\frac{\pi}{6}}\int_{\sqrt{3}d}^{k_\eta}|x-y|\frac{dy
 d\eta}{\frac{\sqrt{3}}{2}d^2}\frac{dx
 d\theta}{\frac{\sqrt{3}}{2}d^2}
\end{eqnarray}
Since $E_w[|x-y|]_3$ is bounded above as follows:
\begin{eqnarray}
E_w[|x-y|]_3&<&2(\int_0^{\frac{\pi}{6}}\int_{\sqrt{3}d}^{k_\theta}\frac{y
dy d\theta}{\frac{\sqrt{3}}{2}d^2})^2
(2-\sqrt{3})\nonumber\\&=&2(1-\frac{\sqrt{3}\pi}{6})^2
(2-\sqrt{3})d\nonumber\end{eqnarray}
Thus, the overall expected
difference between distances to the sensor node is:
\begin{eqnarray}
E[|x-y|]&=&{E_w[|x-y|]}_1+{E_w[|x-y|]}_2+{E_w[|x-y|]}_3\nonumber\\
&<& \frac{\sqrt{3}{\pi}^2 d}{45}+\pi
d(-\frac{2}{3}+\frac{2}{9}\sqrt{3}+\frac{1}{6}\sqrt{3}\ln
3)\nonumber\\&+&2(1-\frac{\sqrt{3}\pi}{6})^2 (2-\sqrt{3})d\nonumber
\end{eqnarray}
Using the above inequality with Eqns.\,\,(\ref{equ:strength
of coverage}) and (\ref{dbar}),
\[
U(G,R) < 0.2038
\]
\end{proof}

\section{}
\label{sec:poisson} In this section of the appendix, we consider
$D(G,R)$ and $U(G,R)$ when the spatial distribution of active nodes
is given by a Poisson point process.
\begin{theorem}
If the spatial distribution of active nodes is given by a Poisson
point process, the expected values of $D(G,R)$ and $U(G,R)$ are
$0.5$ and $1 - 1/\sqrt{2}$, respectively.
\label{theorm:poisson_metric}
\end{theorem}
\begin{proof}
Consider active sensor nodes randomly distributed in the region of
interest, $R$, given by a Poisson process of rate $z$ active nodes
per unit area. Therefore, the probability that we will have $k$
nodes within some area $S$ is given by:
\begin{equation}
\label{poisson-def} P(k,S) = \frac{(zS)^ke^{-zS}}{k!} \end{equation}
In the following, we assume that the region of interest is large
enough to ignore boundary issues.

From Eqn.\,\,(\ref{poisson-def}), the probability that there are $0$
active nodes within a radius of $r$ is given by:
\begin{equation} \label{noNode} P(0, \pi r^2 ) = \frac{(z\pi
r^2)^{0} e^{-z\pi r^2}}{(0)!}=e^{-z\pi r^2}\end{equation}

Consider a ring of radius $r$ of infinitesimal area equal to $2\pi r
dr$. Using Eqn.\,\,(\ref{poisson-def}) again and noting that $2\pi r
z dr \rightarrow 0$ implies $e^{- 2\pi r z dr} \rightarrow 1-2\pi r
zdr$, the probability that there is exactly one active node on this
ring is given by:
\begin{equation}
P(1, 2\pi r dr)  = \frac{ (2\pi r z dr)^1 e^{- 2\pi r z \,dr}}{1!}
\approx 2 \pi r z\,dr \label{ring}
\end{equation}
Thus, the expected distance, $E[X_1]$, between the node and its nearest
neighbor is given by:
\begin{eqnarray}
E[X_1] & = & \int_0^\infty r P(0,\pi r^2) P(1, 2\pi r dr)\nonumber \\
&\approx& \int_0^\infty re^{-z\pi r^2}2\pi rzdr \nonumber \\
& = & \frac{1}{2\sqrt{z}}
\end{eqnarray}

Using Eqn.\,\,(\ref{equ:strength of coverage}), we get:
\[
D(G,R) = \frac{1}{2}
\]

We now derive the expected value of $U(G,R)$. Let $W$ denote a random
variable indicating the distance of a random point from its nearest
node. The probability density function of $W$ is given by:
\[
p_W(r) = 2\pi r e^{-\pi r^2z}
\]
Consider any two random points whose distances to their respective
nearest nodes are $x$ and $y$. Now,
\begin{eqnarray}
E[|x - y|] &=& \int_0^\infty p_W(x) \int_0^\infty p_W(y) |x-y|\,dy\,
dx\nonumber \\
&=&\int_0^\infty p_W(x) \int_0^{x} p_W(y)(x-y)\, dy\, dx \nonumber
\\&+&\int_0^\infty p_W(x)\int_{x}^\infty p_W(y)(y-x)
\,dy\,dx\nonumber
\end{eqnarray}
Focusing first on the inner integrals and simplifying, we get the
following two results:
\begin{eqnarray}
\int_0^{x} p_W(y)(x-y) dy &=& \int_0^{x}e^{-z\pi y^2}2\pi
yz(x-y)dy\nonumber\\
&=&x-\int_0^{x}e^{-z\pi y^2}dy \label{firstInner}
\end{eqnarray}
\begin{eqnarray}
\int_{x}^\infty p_W(y) (y-x) dy &=& \int_{x}^\infty e^{-z\pi
y^2}2\pi yzdy(y-x) dy \nonumber
\\&=&\int_{x}^\infty e^{-z\pi y^2}dy
\label{secondInner}
\end{eqnarray}

Define $g(x)$ as follows:
\[
g(x)=\int_0^{x}e^{-z\pi y^2}dy
\]
Since $g(\infty)=\frac{1}{2}\sqrt{\frac{1}{z}}$,
\[
\int_{x}^\infty e^{-z\pi y^2}dy=\frac{1}{2}\sqrt{\frac{1}{z}}-g(x)
\]
Using (\ref{firstInner}) and (\ref{secondInner}), $E[ |x-y| ]$ may
be expressed as:
\begin{eqnarray}
&~~& \int_0^\infty p_W(x)[x+\frac{1}{2}\sqrt{\frac{1}{z}}-2g(x)] dx\nonumber \\
&\approx&\int_0^\infty e^{-z\pi x^2}2\pi xz \left[
  x+\frac{1}{2}\sqrt{\frac{1}{z}}-2g(x) \right] dx\nonumber\\
&=& \sqrt{\frac{1}{z}}+ 2 e^{-z\pi x^2}g(x)|_0^\infty -
\int_0^\infty 2e^{-z\pi x^2} \frac{d(g(x))}{dx} \,dx \nonumber
\end{eqnarray}
Simplifying further using routine calculus, we get:
\begin{equation}
E[ | x-y | ] = \sqrt{\frac{1}{z}}(1-\sqrt{\frac{1}{2}})
\label{diffavg}
\end{equation}
Given that the expected distance to the nearest node is $1/2\sqrt{z}$,
using Eqn.\,\,(\ref{diffavg}) in the definition of $U(G,R)$, we get:
\[
U(G,R) =
\frac{1}{2\frac{1}{2}\sqrt{\frac{1}{z}}}\sqrt{\frac{1}{z}}(1-\sqrt{\frac{1}{2}})
= 1-\sqrt{\frac{1}{2}} \nonumber
\]
\end{proof}

\bibliographystyle{IEEEtran}
\bibliography{report-EvenRep}
\end{document}